\documentclass[12pt]{article}
\usepackage{amssymb,amsthm,latexsym,amsmath,dsfont}
\usepackage{graphicx}
\graphicspath{{Figures/}}
\usepackage{color}
\usepackage{enumitem}
\numberwithin{equation}{section}
\newcommand{\re}{\mathrm{Re}\, }
\newcommand{\im}{\mathrm{Im}\,}
\newcommand{\I}{{\rm{i}}}
\newcommand{\D}{{\rm{d}}}
\newcommand{\E}{{\rm e}}
\newcommand{\tr}{\operatorname{tr}}
\newcommand{\identity}{{\mathds{1}}}
\newcommand{\Identity}{{\mathrm{Id}}}

\newcommand{\spec}{\sigma}

\newcommand{\be}{\begin{equation}}
\newcommand{\ee}{\end{equation}}
\newcommand{\bea}{\begin{eqnarray}}
\newcommand{\eea}{\end{eqnarray}}
\newcommand{\range}{\operatorname{ran}}
\newcommand{\Span}{\operatorname{span}}

\newcommand{\card}{\operatorname{card}}

\newcommand{\bra}[1]{\langle #1 |}
\newcommand{\braket}[2]{\langle #1 | #2 \rangle}
\newcommand{\ketbra}[2]{| #1 \rangle \langle #2 |}

\def\integer{{\mathbb{Z}}}
\def\pinteger{{\mathbb{N}}}
\def\spinteger{{\mathbb{N}}^\star}
\def\real{{\mathbb{R}}}

\def\complex{{\mathbb{C}}}

\def\proba{{\rm I\kern -.18em P}}

\newcommand{\ie}{i.e.\;}

\newcommand{\RHS}{right-hand side\;\,}

\newcommand{\iv}{{\underline{i}}}
\newcommand{\jv}{{\underline{j}}}
\newcommand{\kv}{{\underline{k}}}
\newcommand{\lv}{{\underline{l}}}

\newcommand{\muv}{{\underline{\mu}}}
\newcommand{\nuv}{\underline{\nu}}

\newcommand{\Aa}{{\cal A}}
\newcommand{\Bb}{{\cal B}}

\newcommand{\Dd}{{\cal D}}

\newcommand{\Ff}{{\cal F}}

\newcommand{\Hh}{{\cal H}}

\newcommand{\Nn}{{\cal N}}

\newcommand{\Pp}{{\cal P}}

\newcommand{\Rr}{{\cal R}}

\newcommand{\Ttt}{{\cal T}}
\newcommand{\Uu}{{\cal U}}
\newcommand{\Vv}{{\cal V}}

\newcommand{\HB}{H_{B}}
\newcommand{\US}{V}
\newcommand{\SB}{S}
\newcommand{\Uint}{U_{\rm int}}
\newcommand{\U}{U_{\lambda}}
\newcommand{\setmu}{\sigma(T)}

\newcounter{resultcounter}[section]

\newtheorem{thm}[resultcounter]{Theorem}
\newtheorem{lem}[resultcounter]{Lemma}
\newtheorem{prop}[resultcounter]{Proposition}
\newtheorem{cor}[resultcounter]{Corollary}

\newtheorem{rmk}[resultcounter]{Remark}
\usepackage{ulem, lipsum}

\begin{document}


\title{Fermionic quantum walkers coupled to a bosonic reservoir}

\author{  Olivier Bourget\footnote{Pontificia Universidad Cat\'olica de Chile, Departamento de Matem\'atica, Santiago, Chile; E-mail: bourget@uc.cl },
  Alain Joye\footnote{Univ. Grenoble Alpes, CNRS, Institut Fourier, 38000 Grenoble, France; E-mail: alain.joye@univ-grenoble-alpes.fr}, and
  Dominique Spehner\footnote{Universidad de Concepci\'on, 
    Departamento de Ingenier\'{\i}a Matem\'atica, Chile {\&}  Univ. Grenoble Alpes, CNRS, Institut Fourier and LPMMC, 38000 Grenoble, France; E-mail: dspehner@ing-mat.udec.cl}
}
\date{ }

\maketitle
		
\begin{abstract}
We analyse the discrete-time dynamics of a model of non-interacting { fermions 
coupled} to an infinite reservoir formed by a bosonic quantum walk on $\integer$.  This dynamics consists of consecutive applications of  free evolutions of the fermions and bosons  followed by a local coupling between them. The unitary operator implementing this coupling 
{ accounts for energy exchanges between the system and reservoir while it preserves the number of fermions.}
The free fermion evolution is given by a second-quantized single-particle unitary operator satisfying some genericity assumptions. 
{ The} free boson evolution is  given by the  second-quantized shift operator on  $\integer$. 
We derive explicitly the Heisenberg dynamics of fermionic observables and obtain a systematic expansion in the large-coupling regime, which we control by using spectral methods. We also prove that the reduced state of the fermions converges in the large-time limit to { a mixture of} infinite-temperature Gibbs { states in each particle sector}.
\end{abstract}
		
\section{Introduction}

Quantum walks (QWs) have emerged as a powerful and versatile modeling tool in quantum physics, with natural applications in areas such as quantum optics and condensed matter physics. They also provide non-trivial models for the mathematical analysis of quantum scattering and spectral theory. Moreover, QWs are used in quantum computing and information processing as building blocks for quantum search algorithms, and they are considered non-commutative extensions of classical random walks and Markov chains in probability theory. For further details, see \cite{Ke, V-A, CC, P, J4, ABJ3, GZ, QMS, Sa, G, APSS1, RT, T, CJWW, J5} and references therein.

In the study of transport phenomena in complex random quantum systems, QWs offer a simplified yet relevant framework for rigorously analyzing the dynamics of quantum particles on infinite discrete structures, such as graphs or lattices (see, e.g. \cite{Ko, JM, ASW, ABJ2, HJS2, J3, HJ, CFGW, CFO}). When statistical ensembles of particles  are considered, fermionic QWs coupled to an environment provide a natural model for studying open quantum systems. The papers \cite{HamzaJoye17, Raquepa20, Anjora21} offer a mathematical description of the onset of an asymptotic state in the case of a single reservoir and describe the particle currents induced by multiple reservoirs in non-equilibrium situations.

In this work, we analyse the discrete-time dynamics of non-interacting fermionic QWs 
 on a finite graph (the sample) coupled to an infinite reservoir of bosonic QWs on the lattice $\integer$. One time step of the dynamics first consists for the fermionic and bosonic quantum walkers to undergo unitary free dynamics on the sample and reservoir, respectively. This is then followed by a unitary coupling between the fermionic and  bosonic { QWs}, acting locally on the { reservoir}. The coupling step allows for energy exchanges while it preserves the number of fermions in the sample. 

{ An example of sample-reservoir coupling we have in mind describes fermionic transitions between two fixed vertices of the graph induced by absorption and emission of  bosons at a single lattice site of $\integer$, the coupling having} 
the form
commonly encountered in quantum optics and condensed matter physics 
to modelize atom-photon and electron-phonon interaction processes.
The free fermionic evolution is given by the second-quantization of a single-particle unitary operator describing hopping of fermions between vertices.     
The free bosonic evolution is  given by the  second-quantization of the shift operator on  $\integer$. 
The initial state of the bosonic QWs is assumed to be a quasi-free state, similar to a thermal state. Unlike previous studies~\cite{HamzaJoye17, Raquepa20, Anjora21}, the distinct natures of the particles in the sample and reservoir prevent us from reducing the problem to a one-body problem. Additionally, while previous works have focused on weak coupling regimes, our framework naturally leads us to explore the large coupling regime, akin to the high-temperature regime, to analyze the asymptotic state of the sample.

In this framework, we derive an explicit description of the Heisenberg dynamics of observables acting on the fermionic sample in terms of the main characteristics of the model. For any fixed time, this expression { yields} a systematic expansion in the large-coupling regime, which we control using spectral methods. Furthermore, under natural generic assumptions on the single-fermion unitary operator describing the free evolution, we prove that the reduced state of the fermionic QWs converges in the large-time limit to a mixture of infinite-temperature Gibbs states in each particle sector. Our approach takes advantage of the local nature of the coupling { on the bosonic lattice $\integer$} to expand the time-evolved fermionic observables as infinite series that can be controlled in the large-coupling limit. 

Our work, albeit devoted to a specific { coupling involving only} one reservoir, belongs to the mathematical analysis of related models characterized by local couplings between fermionic and bosonic particles on graphs or lattices. In particular, a fermionic sample coupled to two bosonic reservoirs in different thermodynamical states, hosting asymptotic currents between reservoirs, can be analysed using a similar approach; this question lies beyond the scope of the present paper and will be addressed in future works.  

The remainder of this paper is organized as follows: The model is described in the next section, while Section 3 sets the notation and recalls the basic properties of Fock spaces that will be used in the following sections. Section 4 is devoted to deriving the expression for the Heisenberg evolution of the sample observables and analyzing its behavior in the large coupling regime. The spectral and long-time properties of the dynamics are presented in Section 5 under genericity assumptions. The paper concludes with Section 6, which justifies some of these assumptions on the basis of unitary random matrices.

\vspace{0,5cm}
\vspace{0,5cm}

\noindent {\bf Acknowledgments:}
{ We thank the referee for suggesting an insightful generalization of our results.} This work has been partially supported by the French National Research Agency in the framework of the ”France 2030” program (ANR-11-LABX-0025-01) for the LabEx PERSYVAL, ANR Dynacqus ANR-24-CE40-5714-02, the Franco-Chilean ECOS-Anid Grant 200035 and the Chilean Fondecyt Grant 1211576.

\section{Model { and main results}} \label{sec-notation}

We consider fermionic particles
{ with a finite-dimensional single-particle Hilbert space $\Hh$ of dimension $d  < \infty$}. 
The fermionic Fock space is the second quantization $\Ff_{-} = \Gamma_{-} (\Hh )$ of $\Hh$,
so that  $\Ff_{-}$ has dimension $2^d$ (see Sec.~\ref{sec-notation} for more detail). 
The fermions do not interact with each other, thus their free dynamics is given by
the second-quantization $\Gamma_{-} (\US)$ of
a unitary operator $\US$ on $\Hh$.  We adopt the convention that a particle on a discrete configuration space defines a QW when the unitary operator giving its one time step evolution only couples sites of the configuration space a finite distance away. Since $\Hh$ { can be identified with $\ell^2 ( \Lambda)$, where $\Lambda$ is a finite graph  with $d= \card (\Lambda)$ sites}, any unitary operator $V$ on $\Hh$ can be considered the evolution of a QW. 

The fermionic particles are coupled to a reservoir composed of
non-interacting bosonic QWs, with Hilbert space given by the bosonic Fock space  $\Ff_{+} = \Gamma_{+} (\Hh_{\rm B}  )$,
$\Hh_{\rm B} = \ell^2 ( \integer )$ being the single-boson space and
$\Gamma_{+}  (\Hh_{\rm B}  )$ its second quantization.
The reservoir's free dynamics is described by the second quantization $\Gamma_{+} ( \SB )$ of
 the shift operator
$\SB$ on  $\Hh_{\rm B}$ defined by
\begin{equation} \label{eq-shift_bosons}
  \SB \delta_j = \delta_{j-1}\;\;,\;\; j \in \integer\;,
\end{equation}
where $\{ \delta_j \}_{j \in \integer}$ is the canonical basis of $\Hh_{\rm B}$.
The reservoir is an example of a statistical ensemble of QWs on a lattice~\cite{HamzaJoye17, Raquepa20, Anjora21}.

The fermion-reservoir coupling is given by a unitary operator $U_{\rm int}$ on $\mathcal{F}_{-} \otimes \mathcal{F}_+$ which has the following form.  We denote 
by $b_j$ and $b_j^\ast$ the bosonic annihilation and creation operators at site $j \in \integer$.
The coupling of the QWs on the sample with the QWs in the reservoir is given by the following unitary operator
on $\Ff_{-} \otimes \Ff_{+}$ 
{ 
\begin{equation} \label{eq-coupling_unitary}
\Uint =  \E^{-\I \lambda T \,\otimes \, ( b_0 + b_0^\ast)/\sqrt{2}}\;, 
\end{equation}    
where $\lambda\in \real$ is a system-reservoir coupling constant, $( b_0 + b_0^\ast)/\sqrt2$ is the bosonic field operator at lattice site $0$, acting on $\Ff_+$, and
$T$ is a self-adjoint operator on $\Ff_{-}$. 
 Note that the coupling Hamiltonian inside the exponential only involves the field operator at site $j=0$ (local coupling) and is linear in the field operator.
We suppose that  
$T= \D \Gamma_{-} ( \tau)$ is the second quantization of a self-adjoint operator $\tau$ on $\Hh$, although our results in Sec.~\ref{sec-large_coupling_limit} 
do not need this assumption. In analogy with the free dynamics, the coupling unitary \eqref{eq-coupling_unitary} is then the second quantization 
$\Uint = \Gamma_{-} \otimes {\rm Id}_{\Ff_+} (V_{\rm int})$ of a unitary operator $V_{\rm int} = \E^{-\I \lambda \tau \otimes (b_0+b_0^\ast)/\sqrt{2}}$ on $\Hh \otimes \Ff_{+}$ (see  Sec.~\ref{sec-notation}).
}

{ Like in} the models of fermionic quantum walks studied in~\cite{HamzaJoye17,Raquepa20,Anjora21}, the total evolution operator of the
fermionic walkers and reservoir is given by a sequential application 
of the  free evolution 
operator $\Gamma_{-} (\US) \otimes \Gamma_{+} ( \SB )$ and the
coupling unitary $U_{\rm int}$,
\begin{equation} \label{eq-total_evolution_op}
  U_\lambda =  U_{\rm int}\, (\Gamma_{-} ( \US )  \otimes \Gamma_{+} ( \SB ))  \;.
\end{equation}
The initial state  is a decorrelated state (product state) 
\begin{equation}
  \omega_{S} \otimes \omega_{B}\;,
\end{equation}
where $\omega_{S}$ and $ \omega_{B}$ are the initial fermionic and bosonic states, respectively. The bosons are
assumed to be { initially in} a quasi-free state (see Sec~\ref{sec-evol_fermionic_observables} for more details).

The state of the fermionic walkers and reservoir is given  at the time step $t\in \pinteger$ by
\begin{equation}
   \omega_S \otimes \omega_B\circ  \Uu_\lambda {  ( \cdot )^t} 
  \quad , \quad
  \Uu_\lambda ( \cdot ) = U^\ast_\lambda \cdot U_\lambda\; ,
\end{equation}
where $\Uu_\lambda ( \cdot )^t$ denotes the $t^{\rm th}$ power of $\Uu_\lambda ( \cdot )$. We are interested in the evolution at time $t$ of observables of the system;
more precisely in the map $\Ttt_t: \Bb(\Ff_-)\rightarrow  \Bb(\Ff_-)$ such that for any $X\in  \Bb(\Ff_-)$, 
\be\label{deftt}
 \omega_S(\Ttt_t(X) )=  \omega_S \otimes \omega_B( \Uu_\lambda^t ( X\otimes \identity ) ).
\ee

{ 
A specific example within the family of models  considered above is a fermionic QW on 
a finite graph $\Lambda$  with $d= \card (\Lambda)$ sites, coupled locally to bosonic QWs by the unitary operator $\Uint$ of the form
\eqref{eq-coupling_unitary} with
\begin{equation} \label{eq-T_transport-model}
T  = T_{\rm{hop}} \equiv \E^{\I \varphi} a_2^\ast a_1 + \E^{-\I \varphi} a_1^\ast a_2\;. 
\end{equation}
Here, $a_l$ and 
$a_l^\ast$ are the fermionic annihilation and creation operators at site $l \in \Lambda$ 
 and
$\varphi$ is a Peierls phase. The latter describes the effect of a magnetic field.
The operator $T$ on $\Ff_{-}$ is the second quantization $\D \Gamma_{-} ( \tau)$ of the single particle self-adjoint operator
$\tau_{\rm{hop}}= \E^{\I \varphi} \ketbra{e_2}{e_1} + \E^{-\I \varphi} \ketbra{e_1}{e_2}$, where 
$e_l$ is the state of a particle localized at site $l \in \Lambda$.
	Note the local character of the interaction at the fermionic sites $1$ and $2$.
	The unitary $U_{\rm int}$ with $T$ given by \eqref{eq-T_transport-model} models the interaction of fermionic atoms with the electromagnetic field
	within the dipolar approximation when the field is only coupled to states $e_1$ and $e_2$, and the $e_l$'s correspond to  atomic energy eigenstates.
	The dipolar approximation, which consists in neglecting non-linear terms in the field operator
	$(b_0+b_0^\ast)/\sqrt{2}$,
	is justified when the size of the system is much smaller than the variation of the electromagnetic field.
}

\begin{rmk}
{ The model with the choice $T=T_{\rm{hop}}$
	is related to, but different from, the 
Jaynes-Cumming model of quantum optics~\cite{Breuer-Petruccione}. In the latter model, in addition to the dipolar approximation one performs a  rotating-wave approximation to remove non-secular terms~\cite{Cohen-Tannoudji}. This 
amounts to drop the terms  $\E^{\I \varphi} \ketbra{e_2}{e_1} \otimes b_0^\ast$ and $\E^{-\I \varphi} \ketbra{e_1}{e_2} \otimes b_0$ in the interaction Hamiltonian. After second quantization, this leads to the coupling unitary 
\begin{equation}  \label{eq-coupling_unitary_RWA}
\Uint^{\rm RWA} =  \E^{-\I \lambda  (  \E^{\I \varphi} a_2^\ast a_1\, \otimes \, b_0 +  \E^{-\I \varphi} a_1^\ast a_2 \, \otimes \, b_0^\ast)/\sqrt{2}}
\;,
\end{equation}    
where the first (respectively second) term inside the parenthesis describes the absorption (resp. emission) of a photon from site $j=0$, provoking the hopping
of a fermion from site $l=1$ to site $l=2$ (resp. a hopping from $l=2$ to $l=1$). 
If one associates to each site $l\in \Lambda$ an energy $E_l$ and assume that $E_2 > E_1$, the evolution operator
$U_{\rm int}^{\rm RWA}$ can be interpreted as
modelling energy exchanges between the sample and reservoir. A similar interpretation can be drawn for  $U_{\rm int}$ in (\ref{eq-coupling_unitary}) with $T=T_{\rm{hop}}$,
albeit one keeps in the Hamiltonian 
terms $\E^{\I \varphi} a_2^\ast a_1 \otimes b_0^\ast$ and  $a_1^\ast a_2 \otimes b_0$ associated to  non-energy conserving processes.
Although the unitaries \eqref{eq-coupling_unitary_RWA} and \eqref{eq-coupling_unitary} describe similar physics, their mathematical structure is different. Our results below can not be straightforwardly extended to system-reservoir couplings of the form \eqref{eq-coupling_unitary_RWA}. 
}
\end{rmk}

{ 
Let us summarize the main results of this paper.
\begin{itemize}
\item[1)] {\it{Expansion of the fermionic evolution operator  $\Ttt_t$ at large coupling constants $\lambda$.}}

\vspace{1mm}

 We show in Sec.~\ref{sec-large_coupling_limit} that the time-evolved fermionic observables $\Ttt_t(X)$ can be written as an infinite series,
the first term of which is equal to $(\Vv \Phi )^t (X)$, where 
$\Vv (\cdot) = \Gamma_{-}(V^{-1}) \cdot \Gamma_{-}(V)$ gives the free evolution of fermionic observables and $\Phi$ is an orthogonal projector on $\Bb ( \Ff_{-})$ for the Hilbert-Schmidt scalar product. The projector $\Phi$ describes the effect on the sample of the coupling with the reservoir in the infinite coupling constant limit $ \lambda \to \infty$. Actually, we show that the remainder of the truncated series
formed by 
the first $s$ terms, with $s \geq 1$, is of order  
$e^{-c\lambda^2(s+1)}$ for some constant $c>0$. 

More precisely, 
we prove  that if the reservoir initial state $\omega_B$ is a quasi-free state (see Sec~\ref{sec-evol_fermionic_observables}) and $T$ is a self-adjoint 
operator on $\Bb(\Ff_{-})$, then there exists $\lambda_0>0$ such that  for all $\lambda > \lambda_0$ and $t \in \spinteger$, the norm of the remainder is bounded from above by
\begin{equation}
C_s \|X\| \,\E^{-\gamma t/2} \E^{-\lambda^2 (s+1)\Delta /4}
\end{equation}  
for some constants $C_s>0$  and $\gamma>0$, where 
$\Delta = \min_{\mu \not= \nu, \mu,\nu \in \spec (T)}|\mu- \nu|$ is the minimal spectral gap of $T$. 
In particular, taking $s=1$, one obtains that
\begin{equation}
\Ttt_t \to (\Vv \Phi)^t \quad \text{uniformly in $t$ as $\lambda \to \infty$\;,}
\end{equation}
where the convergence is in operator norm.  

\item[2)] {\it{Exponential convergence to a steady state in the large time limit.}}

\vspace{1mm}

 We show in Sec.~\ref{sec-convergence_large_time} that the fermionic state converges in the large time limit $t \to \infty$ to a steady state $\omega_S^\infty$  given by an infinite-temperature Gibbs state in each $n$-particle subspace $\Hh^{\wedge n} \subset 
\Ff_{-}$.  

Let us first discuss the assumptions under which we can prove this convergence. In addition 
to the quasi-freeness hypothesis on the reservoir initial state $\omega_B$ and the choice of $T$ as the second-quantized of a self-adjoint operator on $\Hh$ (see above), we make the following hypotheses on the coupling operator $T$ and the fermionic free evolution operator $\Gamma_{-}(V)$. Let $B^\mu$, $\mu \in \spec (T)$, be the spectral projectors  of $T$. We assume that
\begin{itemize}
\item[(i)] $\Vv$ has simple eigenvalues save for the eigenvalue $1$, which is $\dim \Ff_{-}$- fold degenerate.
\item[(ii)] For all $n=1,\ldots, d$, the matrix of the restriction $B^\mu|_{\Hh^{\wedge n}}$ in the eigenbasis of $\Gamma(V) |_{\Hh^{\wedge n}}$ has only non zero diagonal elements, for all $\mu \in \spec (T|_{\Hh^{\wedge n}})$, and for all matrix indices $k,l$, $k \not=l$, at least one spectral projector $B^\mu$ has a nonzero entry with indices $(k,l)$. 
\end{itemize}
Condition (ii) basically means that the spectral projectors of $T$ connect all pairs of eigenstates of $\Gamma_{-} (V)$.
Under these assumptions, we show that
there exists $\lambda_0>0$ such that for $\lambda>\lambda_0$ and $X \in \Bb(\Ff_-)$,
\begin{equation} 
\lim_{t \to \infty} \omega_S ( \Ttt_t ( X)) = \omega_S^\infty ( X)\;,
\end{equation}
where the asymptotic state	$\omega_S^\infty (\cdot )= \tr ( \rho_S^\infty \,  \cdot )$ has a density matrix given by
\begin{equation}
\rho_S^\infty = \bigoplus_{n=0}^d \begin{pmatrix} d \\ n \end{pmatrix}^{-1}\tr_{\Hh^{\wedge n}} ( \rho_S |_{\Hh^{\wedge n}})  \,\identity_{\Hh^{\wedge n}}\;.
\end{equation}
In particular, if there are initially $n$ fermions in the sample, i.e., if the initial state $\omega_S$ has a density matrix $\rho_{S}$ with support and range  in $\Hh^{\wedge n}$,
then the asymptotic state is an infinite-temperature Gibbs state in the $n$-particle sector,
\begin{equation}
\rho_S^\infty= \begin{pmatrix} d \\ n \end{pmatrix}^{-1} \,\identity_{\Hh^{\wedge n}}\;,
\end{equation}
where, with a slight abuse of notation, $\identity_{\Hh^{\wedge n}}$ stands for the 
 orthogonal projector onto $\Hh^{\wedge n}$.
\end{itemize}

}

\section{Notation} 

In this section we briefly specify the notation used and recall the definitions and basic properties of the fermionic and bosonic creation and annihilation operators. 

Given a Hilbert space $\Hh$ of dimension $d<\infty$, one constructs the associated 
fermionic Fock space $\Ff_{-}$ as follows (see e.g.~\cite{Bratteli} for more detail).
Let $\Hh^{\otimes n}$ denote the $n$-fold tensor product of $\Hh$. The antisymmetric tensor product of the
 vectors $u_1, \dots, u_n\in \Hh$, $1 \leq n\leq d$, reads  
\be\label{wedge}
u_1\wedge  \cdots \wedge u_n=\frac{1}{\sqrt{n!}}\sum_{\pi\in S_n} \epsilon_{\pi}  u_{\pi(1)}\otimes \dots \otimes u_{\pi(n)} \ \in\ \Hh^{\otimes n}\;,
\ee
where $S_n$ is the  group of permutations of $\{1,\dots ,n\}$ and $\epsilon_{\pi}$ is the signature of $\pi \in S_n$. The scalar product of two such vectors is given by
\begin{align}\label{scalprodet}
\langle u_1\wedge  \cdots \wedge u_n | v_1\wedge  \cdots \wedge v_n \rangle= \det\big(  (\langle u_k | v_l\rangle)_{1\leq k, l  \leq n} \big)\;,
\end{align}
where $\braket{u}{v}$ is the scalar product of $u$ and $v$ in $\Hh$.
For $1 \leq n \leq d$, the orthogonal projection $\Pp_{-}^{(n)}$ acting on $\Hh^{\otimes n}$  is defined by linearity and its action on the $n$-fold product vectors, 
\be\label{defproj}
\Pp_{-}^{(n)}  u_1\otimes   \cdots \otimes u_n = \frac{1}{\sqrt{n!}} \, u_1\wedge  \cdots \wedge u_n\;.
\ee
The $n$-fold antisymmetric tensor product of $\Hh$, denoted by $\Hh^{\wedge n}$, is the subspace $\Pp_{-}^{(n)} \Hh^{\otimes n}$.

If {  $\{ f_i \}_{i=1}^d$ is a basis of $\Hh$,
then $\Hh^{\wedge n}$ is generated by the vectors $f_{i_1} \wedge f_{i_2} \wedge \dots \wedge f_{i_n}$, $\iv= (i_1,i_2,\dots, i_n) \in I_n$, where} 
\be \label{eq-notation_I_n}
I_n=\bigl\{ \iv \in \pinteger^n\,;\, 1\leq i_1 < i_2 < \dots < i_n\leq d\big\}
\ee
is the  set of ordered indices. The fact that one can restrict $\iv$ to this set follows from the antisymmetry of $f_{i_1} \wedge \dots \wedge f_{i_n}$ under the permutation of two vectors.
By \eqref{scalprodet}, if the basis $\{ f_i \}_{i=1}^d$ is orthonormal then $\{f_{i_1}\wedge \dots \wedge f_{i_{n}} \}_{\iv \in I_n}$ is an orthonormal basis of $\Hh^{\wedge n}$. Thus
$\Hh^{\wedge n}$  has dimension
$\begin{pmatrix}d \\ n \end{pmatrix}$.
For short, we will denote in what follows the basis vectors by

\be\label{shortnotwed}
f_{i_1}\wedge f_{i_2}\wedge \dots \wedge f_{i_{n}}=\wedge^n f_{\iv} \quad , \quad  \iv \in I_n\;.
\ee
The antisymmetric Fock space is then defined as
\be
\Ff_{-}=\bigoplus_{n=0}^{d} \Hh^{\wedge n}  \ \text{ with } \   \Hh^{\wedge 0}=\mathbb{C} \Omega\; , 
\ee
where  $\Omega$ denotes the vacuum vector. { Thus}, $\Ff_{-}$ has dimension $2^d$.

For any $ \phi\in \Hh$, the fermionic creation operator $a^* (  \phi)$ on $\Ff_{-}$ is defined by linearity and its action 
on  $\Omega $ and on $u_1\wedge \dots \wedge u_n\in\Hh^{\wedge n}$, $1 \leq n \leq d$, as 
\be\label{vecreator}
a^*( \phi)\Omega =  \phi \,, \quad a^*( \phi)u_1\wedge  \cdots \wedge u_n= \phi\wedge u_1\wedge  \cdots \wedge u_n \;.
\ee
The creation operator maps $\Hh^{\wedge n}$ onto $\Hh^{\wedge (n+1)}$ and satisfies $a^*( \phi)^2=0$. Note that 
\be\label{wedgeofa}
u_1\wedge \cdots  \wedge u_n=a^*(u_1)\cdots a^*(u_n) \Omega\;.
\ee
The adjoint $a( \phi)$ of $a^*( \phi)$ satisfies $a( \phi) \Omega=0$ and
\be
a( \phi)u_1\wedge  \cdots \wedge u_n=\sum_{j=1}^n(-1)^{j-1}\braket { \phi}{u_j}  u_1\wedge  \cdots \wedge u_{j-1}\wedge u_{j+1}\wedge \cdots \wedge u_n.
\ee
The creation and annihilation operators satisfy the Canonical Anticommutation Relations (CAR): for any $ \phi, \psi \in \Hh$, 
\be \label{eq-CAR}
\{a( \phi),a(\psi)\}=\{a^*( \phi),a^*(\psi)\}=0, \ \ \ \{a( \phi),a^*(\psi)\}=\braket  \phi \psi \identity\;,
\ee
where $\{ \cdot , \cdot\}$ denotes the anticommutator.
The number operator $n( \phi)$ is the self-adjoint operator 
\begin{align} \label{eq-def_number_op}
n( \phi)=a^*( \phi)a( \phi)\;.
\end{align}
Using \eqref{eq-CAR} one finds that  $n( \phi)^2=\| \phi \|^2 \,n( \phi)$.
Hence  if $\phi$ is normalized then $n(\phi)$ is an orthogonal projector.
For a fixed orthonormal basis { $\{ f_i\}_{i=1}^d$ of $\Hh$, the simpler notations $a^\ast_j=a^\ast(f_j)$, $a_j=a (f_j)$ and $n_j=n(f_j)$} will be used in the sequel.

The construction of the bosonic Fock space $\Ff_{+}$ associated to the single-particle Hilbert space $\Hh_B$ proceeds analogously, replacing the projectors \eqref{defproj}
by the projectors $\Pp_{+}^{(n)}$ onto the symmetrized vectors
\be\label{defprojsym}
\Pp_{+}^{(n)}  u_1\otimes   \cdots \otimes u_n =\frac{1}{n!}\sum_{\pi\in S_n} u_{\pi(1)}\otimes \cdots\otimes u_{\pi(n)}
\ee
with $n \in \pinteger$, $n \geq 1$.
Then
\be
\Ff_{+} =\bigoplus_{n=0}^{\infty} \Pp_{+}^{(n)}  \Hh^{\otimes n} \ \text{ with } \  \Pp_{+}^{(0)}\Hh^{\otimes 0}=\mathbb{C} \Omega \;.
\ee
For any $f \in \Hh_B$, the bosonic creation operator $b^\ast ( f)$ is defined similarly to \eqref{vecreator} as
\be
b^\ast ( f) \Pp_{+}^{(n)}  \phi_1\otimes   \cdots \otimes \phi_n = { \sqrt{n+1}}\,\Pp_{+}^{(n+1)}  f \otimes \phi_1\otimes   \cdots \otimes \phi_n\;. 
\ee
Its adjoint is denoted by $b(f)$. Unlike in the fermionic case, $b^\ast(f)^2$ does not vanish.
The creation and annihilation operators  $b^\ast ( f)$ and  $b ( f)$  satisfy the Canonical Commutation Relations (CCR)
\be
\big[ b(f) ,  b(g)\big] =\big[ b^*(f) ,  b^*(g)\big] =0, \ \ \ \big[ b(f)  ,  b^*(g) \big] =\braket{f}{g}\ \identity
\ee
for any $f,g \in \Hh_B$.
The operators  $b^\ast ( f)$ and  $b ( f)$ are unbounded and it is convenient to work with 
unitary Weyl operators defined by
\be \label{eq-def_Weyl_op}
W(f) = \E^{\frac{\I}{\sqrt{2}} ( b(f) + b^\ast(f) )} \ , \quad f \in \Hh_B\;,
\ee  
which inherit from the CCR the composition law
\begin{equation} \label{eq-product_Weyl_op}
  W (f) W (g ) = W (f+g) \E^{-\frac{\I}{2} \im \braket{f}{g}} .
\end{equation}  

Let us recall the action of the second quantizations $\Gamma_\pm (A)$ of  $A$, a contraction on $\Hh$, i.e. $\|A\|\leq 1$, on the fermionic, respectively bosonic, Fock spaces { (see~\cite{Bratteli} for more detail).   They are defined by $\Gamma_\pm (A) \Omega = \Omega$ and} 
\begin{align} \label{+eq=definition_second_quantizationGamma}
&\Gamma_-(A)\, u_1\wedge \cdots  \wedge u_n = Au_1\wedge \cdots  \wedge Au_n\,, \quad  { n=1,\ldots, d} \nonumber\\
&\Gamma_+(A) \Pp_{+}^{(n)}  \phi_1\otimes   \cdots \otimes \phi_n =  \Pp_{+}^{(n)}  A\phi_1\otimes   \cdots \otimes A\phi_n\;,\quad { n \geq 1}\; .
\end{align}
{ If $U$ is a unitary on $\Hh$, then $\Gamma_{\pm} (U)$ is unitary on $\Ff_{\pm}$.
	Given a self-adjoint  operator $H$ on $\Hh$, its second quantization 
	$\D \Gamma_{-} ( H)$ is defined on the fermionic Fock space by $\D \Gamma_{-} ( H) \Omega = 0$ and, for any 	$n=1,\ldots, d$,
	\begin{equation} \label{eq-second_quantized_self_adjoint}
	\D \Gamma_{-} (H)\, u_1 \wedge \cdots \wedge u_n = \sum_{m=1}^n u_1\wedge  \cdots \wedge H u_m \wedge \cdots \wedge u_n \;.
	\end{equation}
	One can define similarly $\D \Gamma_{+} (H)$ on $\Ff_{+}$. 
 In particular, 
	$\D \Gamma_{\pm} ( \identity )= N$ is the particle number operator (in the fermionic case,  $N= \sum_{i=1}^d n_i$).
One has $\Gamma_{\pm} (\E^{\I\, H} )=\E^{\I\, \D \Gamma_{\pm} (H)}$.  
The second-quantized operator $\D \Gamma_{-} (H)$ is quadratic in the annihilation and creation operators,
\be \label{eq-second_quantized_dGamma}
\D \Gamma_{-} (H) = \sum_{i,j=1}^d \bra{f_i} H f_j \rangle \, a_i^\ast a_j\;,
\ee
with a similar expression for $\D \Gamma_{+} (H)$ in terms of $b_i^\ast$ and $b_j$. 
}
Note also the Bogoliubov relations which will be used in the sequel. For $U$ unitary on $\Hh$ and $\phi\in \Hh$,
\be\label{bogrel}
\Gamma_-(U)a^\#(\phi)\Gamma_-(U^*)=a^\#(U\phi)), \ \ \Gamma_+(U)b^\#(\phi)\Gamma_+(U^*)=b^\#(U\phi)\;,
\ee
where $a^\#$ and $b^\#$ stand for creation or annihilation operators.

\section{Fermionic dynamics in the large coupling limit} \label{sec-large_coupling_limit}

\subsection{Evolution of the fermionic observables} \label{sec-evol_fermionic_observables}

{
	
Let us introduce the spectral decomposition of the coupling operator $T$,
\be	
T = \sum_{\mu \in \spec (T)} \mu \, B^\mu\;,
\ee
where $B^\mu$ is its spectral projector for the eigenvalue  $ \mu \in \spec (T)$. 
Then \eqref{eq-coupling_unitary} can be rewritten as
\begin{align} \label{eq-U_int_sum_B}
 U_{\rm int} &  = \E^{-\I \lambda T \,\otimes \, ( b_0 + b_0^\ast)/\sqrt{2}}
\nonumber \\
& = \prod_{\mu \in \setmu} \E^{-\I \lambda  \mu B^\mu \otimes ( b_0 + b_0^\ast)/\sqrt{2}}  
=
\prod_{\mu \in \setmu} \Big( ( \identity - B^\mu)\otimes \identity_{\Ff_{+}} 
+ B^\mu \otimes \E^{-\I \lambda   \mu ( b_0 + b_0^\ast)/\sqrt{2}}\Big)
\nonumber
\\
& = 
\sum_{\mu \in \setmu} B^\mu \otimes W^{-\lambda \mu}\;,
\end{align}
where  $W^\alpha = W(\alpha \delta_0)= \E^{\I \alpha ( b_0 + b_0^\ast )/\sqrt{2}}$ is the bosonic  Weyl operator at site $i=0$,
see \eqref{eq-def_Weyl_op}, 

}

Using the properties { $\Gamma_\pm ( U_1) \Gamma_\pm(U_2)= \Gamma_\pm ( U_1 U_2 )$}  of the second quantization,  
 we have from (\ref{eq-total_evolution_op}), for all $t\in \mathbb N$
\begin{align} \label{eq-total_evol_op}
  U_\lambda^t
&  =  
\Gamma_{-}( \US^t ) \otimes \Gamma_{+} ( \SB^t )
  \Bigl( \Gamma_{-}( \US^{-t} ) \otimes \Gamma_{+} ( \SB^{-t} ) U_{\rm int}  \Gamma_{-}( \US^{t} ) \otimes \Gamma_{+} ( \SB^{t} ) \Big) \cdots 
   \\   \nonumber
   &   
   \cdots  \Bigl( \Gamma_{-}( \US^{-2} ) \otimes \Gamma_{+} ( \SB^{-2} ) U_{\rm int}  \Gamma_{-}( \US^{2} ) \otimes \Gamma_{+} ( \SB^{2} ) \Big)
   \Bigl( \Gamma_{-}( \US^{-1} ) \otimes \Gamma_{+} ( \SB^{-1} ) U_{\rm int}  \Gamma_{-}( \US^{1} ) \otimes \Gamma_{+} ( \SB^{1} ) \Big)\;.
\end{align}
In view of (\ref{eq-shift_bosons}), (\ref{eq-def_Weyl_op}) and (\ref{bogrel}) one has
\begin{align} \label{eq-free_evol_Weyl}
  \nonumber
    W_{j}^{\lambda}  :=  \Gamma_{+} ( \SB^{-j} ) W^{\lambda} \Gamma_{+} ( \SB^{j} ) 
    = & \; \exp \Big\{ \I \lambda \Gamma_{+} ( \SB^{-j} ) \frac{b_0 + b_0^\ast}{\sqrt{2}} \Gamma_{+} ( \SB^{j} ) \Big\} \\
    = & \; \E^{\I \lambda ( b_j + b_j^\ast)/\sqrt{2}} \;,
 \end{align}   
where $b_j = b( \delta_j)$ and $b_j^\ast = b^\ast ( \delta_j)$ are the annihilation and creation operators of a boson at site $j\in\integer$ of the reservoir lattice.
 Collecting (\ref{eq-U_int_sum_B}), (\ref{eq-total_evol_op}), 
 (\ref{eq-free_evol_Weyl}), and introducing the operator $\Vv$ on $\Bb( \Ff_{-})$ defined by
\begin{align}\label{defvcal}
\Vv(X)=\Gamma_{-}(V^{-1}) X  \Gamma_{-}(V),  \ \ \forall \, X\in \Bb(\Ff_{-})\;
\end{align}
and the operators on $\Ff_{-}$ defined by
\begin{equation} \label{eq-expression_B_j}
B^{\mu}_j  =  \Vv^j(B^{\mu} )  
 , \quad j \in \integer ,\; \mu \in \setmu ,
\end{equation}
one gets
\begin{equation} \label{eq-evolution_op_with_Weyl_product}
U_\lambda^t = U_0^t \sum_{\mu_1,\ldots,\mu_t \in \setmu} B_t^{\mu_t} \cdots B_1^{\mu_1} \otimes W_{t}^{-\mu_t \lambda} \cdots W_{1}^{-\mu_1 \lambda}\; .
\end{equation}

Since $W_{j}^{-\mu_j \lambda} = W ( -\mu_j\lambda  \,\delta_{j})$,
where the Weyl operators $W ( f)$ satisfy \eqref{eq-product_Weyl_op}, 
one has
\begin{equation} \label{eq-multiplication_W_j}
  W_{t}^{-\mu_t \lambda} \cdots W_{1}^{-\mu_1 \lambda} = W \Bigl( -\lambda \sum_{j=1}^t \mu_j \delta_{j} \Big)\;.
\end{equation}  
%
\begin{rmk} \label{rmk-1}  
We use hereafter the fact that the shifted states $\SB^{-j} \delta_0 = \delta_{j}$ form an orthonormal family of $\ell^2 (\integer)$. 
For more general coupling unitaries $\Uint$ obtained by replacing $b_0$ and $b_0^\ast$ by $b_\psi$ and $b_\psi^\ast$
with $\psi \in \ell^2 (\integer)$ not
localized on a single site, the $\psi_j = \SB^{-j} \psi$ are not any longer orthogonal. As a result, some additional complicated  phase factors
appear in (\ref{eq-multiplication_W_j}).
\end{rmk}

We can now determine the evolution of arbitrary fermionic observables of the system. Such observables are polynomials in the annihilation
and creation operators, with monomials $a_{l_1}^{\sharp_1}\cdots a_{l_p}^{\sharp_p}$. We write such observables as $X$, for short, the evolution of which is given  at time step $t$ by
\begin{eqnarray}
  \Uu_\lambda^t ( X \otimes \identity ) & = & (\U^\ast)^t  (X\otimes \identity) \, \U^t
  \\   \nonumber
  & = &
  \sum_{\muv,\nuv \in \setmu^t } B_1^{\mu_1} \cdots B_t^{\mu_t}  \, \Vv^t(  X )
  \, B_{t}^{\nu_t} \cdots B_1^{\nu_1}
  \otimes  W \Big(  \lambda \sum_{j=1}^t (\mu_j-\nu_j) \delta_{j} \Big)\;,
\end{eqnarray}
where we have used (\ref{eq-evolution_op_with_Weyl_product}), (\ref{eq-product_Weyl_op}), and
(\ref{eq-multiplication_W_j}) and the sum is over all $\muv=(\mu_1,\ldots,\mu_t)\in \setmu^t$ and $\nuv= (\nu_1,\ldots, \nu_t) \in \setmu^t$.
Note that no phase factor appears
since the scalar product $\langle \sum_{j} \mu_j \delta_{j}\,,\,\sum_l \nu_l \delta_{l} \rangle =  \muv \cdot \nuv$ is real (see remark~\ref{rmk-1}).

The initial state $\omega_B$ of the reservoir we consider is defined as a gauge-invariant quasi-free state of the following general form. Let $K=K^*$ be a bounded operator on $\Hh_{B}=l^2(\integer)$
such that
\begin{align}\label{k>1}
K\geq \identity.
\end{align}
The initial state is characterized by its action on the Weyl operators
\begin{equation} \label{eq-initial_state_boson_action_on_Weyl}
  \omega_{B} ( W (f) ) = \exp \Big( - \frac{1}{4} \big\langle f \, , K  f \big\rangle \Big)
\; , \quad f \in \Hh_{B}
  \; ,
\end{equation}
which defines a {\it bona fide} state, see \cite{Pe},  Chapter 3. 

\begin{rmk} \label{rmk_Gibbs_state}
  In the thermal case, $K=K_\beta=\coth ( {\beta (\HB-\mu\identity)}/{2} ) $, where $\HB=\HB^*$ is a Hamiltonian, $\beta >0$ is the inverse temperature and
   the chemical potential $\mu\in \real$ is chosen so that $\HB-\mu\identity>0$. In this case, $\omega_B$ is a Gibbs state. 
 \end{rmk}
Thus the time evolved observable $X\in \Bb(\Ff_-)$  at time $t$ defined in (\ref{deftt}) reads
\begin{equation} \label{eq-def_Ttt_m}
  \Ttt_t ( X)
  = \sum_{\muv,\nuv \in \setmu^t } B_1^{\mu_1} \ldots B_t^{\mu_t} \Vv^t(X) B_{t}^{\nu_t} \ldots B_1^{\nu_1} 
 \exp \bigg\{ -\frac{\lambda^2}{4} \Big\langle \Theta^{\muv-\nuv}_t \, , K\, \Theta^{\muv-\nuv}_t \Big\rangle \bigg\} , 
\end{equation}
where
\begin{equation} \label{eq-Sigma}
\Theta^{\muv-\nuv}_t = \sum_{j=1}^t  ( \mu_j - \nu_j ) \delta_{j}\;.
\end{equation}
%
\begin{lem}\label{l44}
	For any $t\in \mathbb N$, the map $\Ttt_t : \Bb(\Ff_-)\rightarrow \Bb(\Ff_-)$ is completely positive, unital and trace preserving.
\end{lem}

\begin{proof}
	The map $\Ttt_t$ is the composition of $X\mapsto X\otimes \identity$ from $\Bb(\Ff_-)$ to $\Bb(\Ff_-)\otimes \Bb(\Ff_+)$, $Y\mapsto  (\U^\ast)^t \, Y \, \U^t$ on $ \Bb(\Ff_-)\otimes \Bb(\Ff_+)$, and $Y\mapsto \omega_{B}(Y)$ from $\Bb(\Ff_-)\otimes \Bb(\Ff_+)$ to $\Bb(\Ff_-)$, where $\Bb(\Ff_-)$ is finite dimensional and $\Bb(\Ff_+)$ is separable. It is thus immediate to see that the identity is preserved. Then, the first map above is completely positive since it preserves the spectrum, while the second one is completely positive thanks to Kraus { Theorem~\cite{Kraus70}}. It remains to show { that} the third map is completely positive to get that $\Ttt_t$ is completely positive.
	
	Since $\Bb(\Ff_-)$ is finite dimensional, the complete positivity of the third map amounts to showing that
	$\omega_{B}\otimes \identity_n$ from $\Bb(\Ff_+)\otimes M_n(\mathbb C)$ to $M_n(\mathbb C)$ is positive for all $n\in \mathbb N$, { $n \geq 1$}. Any  $B\in \Bb(\Ff_+)\otimes M_n(\mathbb C)$ can be viewed as a  matrix $(B_{i j})_{1\leq i,j\leq n}$ with $B_{ij}\in  \Bb(\Ff_+)$, $\forall\,i,j\in \{1, \dots, n\}$, so that $\omega_{B}\otimes \identity_n(B)=(\omega_{B}(B_{i j})_{1\leq i,j\leq n}\in M_n(\mathbb C)$. The positivity of the latter amounts to having $\omega_B(\sum_{i,j}\bar\beta_i \beta_j B_{i j})\geq 0$ for any { $\{\beta_j\}_{1 \leq j\leq n}\in \mathbb C^n$}, which holds true if $\sum_{i,j}\bar\beta_i \beta_j B_{i j}$ is a positive element of $\Bb(\Ff_+)$, since $\omega_B$ is a state. But $(B_{i j})_{1\leq i,j\leq n}$ is positive in $\Bb(\Ff_+)\otimes M_n(\mathbb C)$ if for all $\{\varphi_j\}_{1\leq j\leq n}$, $\varphi_j\in \Ff_+$, $\sum_{i,j}\langle \varphi_i | B_{i j} \varphi_j\rangle \geq 0$. In particular, 
	$\varphi_j=\beta_j \varphi$ for some $\varphi\in \Ff_+$ yields $\sum_{i,j}\langle \varphi |  \bar \beta_i \beta_j B_{i j}  \varphi\rangle \geq 0$ for any $\varphi\in \Ff_+$. Hence, $\omega_{B}\otimes \identity_n$ is positive.

	Finally, that fact that $\Ttt_t$ is trace preserving is a consequence of the cyclicity of the trace, the unitarity of $\Gamma_-(V)$, and { the identities $B^\mu B^\nu = \delta_{\mu,\nu} B^\mu$ and $\sum_{\mu} B^\mu = \identity$ of the spectral projectors $B^\mu$}, which make the exponential factor in \eqref{eq-def_Ttt_m} disappear.
\end{proof}

\begin{rmk} 
	If $X = a_{l_1}^{\sharp_1}\cdots a_{l_p}^{\sharp_p}$ , we have, using (\ref{bogrel}), $\Vv^t (  X )= a^{\sharp_1} ({\phi_{l_1}^t})\ldots a^{\sharp_p} ({\phi_{l_p}^t})$, where 
	\begin{equation}
		\phi_l^j = \US^{-j} { f_l}  \quad , \quad l \in \Lambda \;.
	\end{equation}
\end{rmk}

Recalling (\ref{defvcal}) and (\ref{eq-expression_B_j}) we deduce the following result from (\ref{eq-def_Ttt_m}).

\begin{prop}
The time-evolved fermionic observables are given at time $t \in \pinteger$ by
\begin{align}\label{dectvb}
 \Ttt_t ( X)
  = \sum_{\muv,\nuv \in \setmu^t } \Vv \Bb^{\mu_1 \nu_1} \Vv \Bb^{\mu_2 \nu_2} \ldots \Vv \Bb^{\mu_t \nu_t} (X) 
 \exp \bigg\{ -\frac{\lambda^2}{4} \bigg\langle \Theta^{\muv-\nuv}_t \, , K\, \Theta^{\muv-\nuv}_t \Big\rangle \bigg\}\;,
\end{align}
where $\Vv$ is the free fermion evolution operator on $\Bb(\Ff_{-})$, see \eqref{defvcal},
and the operators $\Bb^{\mu \nu}$ on $\Bb(\Ff_{-})$, with $\mu, \nu \in \setmu$,
{  are defined in terms of the spectral projectors  $B^\mu$} by 
\begin{align}\label{defbmunu}
\Bb^{\mu \nu}(X)=B^\mu X B^\nu, \ \ \forall X\,\in \Bb(\Ff_{-})\;.
\end{align}
\end{prop}

It is worth pointing out that  for all non-diagonal terms
$\muv \not= \nuv$,
the exponential factor  in the double sum \eqref{eq-def_Ttt_m} decreases as the coupling strength $\lambda$ increases.
In the strong coupling limit $\lambda \to \infty$,
this exponential factor vanishes save for the diagonal terms $\muv = \nuv$,
for which $\Theta_t^{\muv-\nuv}=\Theta_t^{\underline{0}}=0$.
As we shall show in the following section, it follows from this observation that
\begin{equation} \label{eq-evolution_strong_coupling}
  \lim_{\lambda \to \infty}
  \Ttt_t ( X )
  =  (\Vv\Phi)^t(X)
\end{equation}
uniformly in  time $t \in \pinteger$, where 
\begin{equation} \label{eq-def_CPmapPhi}
  \Phi(\cdot) = \sum_{\mu \in \setmu} \Bb^{\mu \mu}(\cdot) = \sum_{\mu\in \setmu} B^\mu \cdot B^\mu\; .
\end{equation}
By Kraus Theorem, the maps $\Phi$ and $\Vv$ on $\Bb(\Ff_{-})$ are completely positive and trace preserving (CPTP).
Thus the evolution is given in the strong coupling limit by powers of the CPTP map $\Vv\Phi$.
Note that this map is unital, since it satisfies $\Vv\Phi ( \identity ) = \identity$, because $\Phi(\identity)=\sum_{\nu \in \setmu} (B^\nu)^2 = \identity$. We shall show { in Sec.~\ref{sec-convergence_large_time}} that  under suitable hypotheses on the fermionic evolution operator
$\Gamma(V)$, the operators  invariant under 
$\Vv\Phi$ are proportional to $\identity$ in each $n$-particle subspace.

We also point out that when the reservoir's initial state $\omega_B$ is a Gibbs state with inverse temperature $\beta$ and bounded Hamiltonian $H_B$ (see Remark~\ref{rmk_Gibbs_state}),
the time-evolved observable $\Ttt_t ( X )$ also reduces to the \RHS of \eqref{eq-evolution_strong_coupling} in the high temperature limit
$\beta \to 0$, for fixed coupling constants $\lambda< \infty$ (in fact, $\| K_\beta\| \to \infty$ as $\beta \to 0$).

In the opposite weak coupling limit $\lambda \to 0$, the off-diagonal terms 
make a significant contribution in the expansion \eqref{dectvb}. For $\lambda=0$, the exponential factor is equal to $1$ and  the series reduces to $\Vv^t (X)$.
 
 \medskip
 
\medskip 

\noindent{\bf Repeated Interaction Case:}

\medskip

In the particular case where the state of the reservoir is characterized by a symbol $K$ which is diagonal in the canonical basis $\{\delta_j\}_{j\in \mathbb{Z}}$, our system reduces to a Repeated Interaction System, also known as Collisional Model, see {\it e.g} \cite{BJM}. Indeed, this property makes the correlations between QWs in the reservoir vanish, so that the dynamics on the sample reduces to the composition of independent maps:

\begin{cor}
	Assume $K\delta_j= { k_j} \delta_j$, $j\in \mathbb{Z}$, and set { $\kappa_j=e^{-\frac{\lambda^2}{4} k_j}$}. For any $X\in \Bb(\Ff_{-})$ it holds
	\begin{equation}\label{RIS}
		\Ttt_t ( X)
		= \Vv \mathfrak{B}_1 \Vv\mathfrak{B}_2 \ldots \Vv \mathfrak{B}_t (X),
	\end{equation}
	where
	\begin{equation}
	\mathfrak{B}_j(X)=\sum_{\mu,\nu \in \setmu }\kappa_j^{(\mu-\nu)^2}\Bb^{\mu \nu}(X), \ \ \forall\, X\in \Bb(\Ff_{-}).
	\end{equation}
\end{cor}

\begin{proof}
	 The spectral assumption on $K$ makes the correlations within the reservoir factorize,
	\begin{equation}
\exp \bigg\{ -\frac{\lambda^2}{4} \Big\langle \Theta^{\muv-\nuv}_t \, , K\, \Theta^{\muv-\nuv}_t \Big\rangle \bigg\} 
= { \prod_{j=1}^t e^{-(\mu_j-\nu_j)^2 k_j\frac{\lambda^2}{4}} \;,}
	\end{equation}
	which yields the result when inserted in \eqref{dectvb}. 
\end{proof}

The maps $\mathfrak{B}_j$ are parameterized by the eigenvalues { $k_j\geq 1$}, and the order in which they appear in the composition (\ref{RIS}) is dictated by the order of the eigenvectors $\delta_j$ in the reservoir. Also, { $\mathfrak{B}_j$ is unital and CPTP}  for each $j\in \mathbb Z$, as a consequence of Lemma \ref{l44}.

\subsection{Expansion of the propagator $\Ttt_t$  at large coupling} \label{sec-expansion_propagator_large_coupling}

As explained above, the expansion \eqref{dectvb} is well-suited to study the strong coupling limit $\lambda \to \infty$.
Since we assume $K\geq 1$, the exponential factor in (\ref{dectvb}) is bounded by
\begin{align}\label{decayexpterm}
  \exp \bigg\{ -\frac{\lambda^2}{4} \Big\langle \Theta^{\muv-\nuv}_t \, ,  \,  K\, \Theta^{\muv-\nuv}_t \Big\rangle \bigg\}\leq
  \E^{-\lambda^2 \big\|\Theta^{\muv-\nuv}_t\big\|^2/4}= \E^{-\lambda^2 \sum_{1\leq j\leq t}(\mu_j-\nu_j)^2/4}\;,
\end{align}
{ see \eqref{eq-Sigma}.}

 Let us introduce the parameters
\be
\kappa_j = \exp \Big\{ - \frac{\lambda^2}{4}  \big\langle \delta_{j} \,,\,K\,\delta_{j} \big\rangle \Big\}\;.
\ee
Note that $\kappa_j \to 0$ as $\lambda \to \infty$ uniformly in $j \in \integer$.

\vspace{1mm}
\begin{prop}
	We have
\begin{align}\label{dectvbnums}
\nonumber  \Ttt_t
  &= (\Vv\Phi)^t \\
  &\hspace*{5mm} + \sum_{s=1}^t\;\; \sum_{1 \leq j_1 < j_2 < \cdots < j_s \leq t} \;\; \sum_{{\mu_{j_1}\not= \nu_{j_1}, \dots, \mu_{j_s}\not= \nu_{j_s} }\atop {\nu_{j_i}, \mu_{j_i} \in \setmu}}
(\Vv\Phi)^{j_{1}-1}\Vv\Bb^{\mu_{j_1}\nu_{j_1}}(\Vv\Phi)^{j_2-j_{1}-1}\Vv\Bb^{\mu_{j_2}\nu_{j_2}}\nonumber \\ 
& \hspace{1.4cm}\cdots (\Vv\Phi)^{j_s-j_{s-1}-1}\Vv\Bb^{\mu_{j_s}\nu_{j_s}}(\Vv\Phi)^{t-j_s} 
\exp \bigg\{ -\frac{\lambda^2}{4} \sum_{p, q=1}^s ( \mu_{j_p}-\nu_{j_p}) 
(\mu_{j_q} - \nu_{j_q}) \big\langle \delta_{j_p}\, , K\, \delta_{j_q} \big\rangle \bigg\}\\
\label{leadingorder}
&= (\Vv\Phi)^t + \sum_{j=1}^t \sum_{{\mu_j,\nu_j \in \setmu} \atop {\mu_j \not= \nu_j}}(\Vv\Phi)^{j-1}\Vv\Bb^{\mu_{j}\nu_{j}}(\Vv\Phi)^{t-j}  \kappa_j^{(\mu_j-\nu_j)^2}
+\; \Rr_t \;.
\end{align}
\end{prop}

\begin{proof}
	 Decompose the double sum over $\muv$ and $\nuv$ in (\ref{dectvb}) according to the number of indices $1\leq j\leq t$ where $\mu_j$ and $\nu_j$ differ.
\end{proof}

To estimate the remainder term $ \Rr_t$, we use the decay properties of the powers of $\Vv\Phi$ appearing in \eqref{dectvbnums} in the subspace orthogonal to their
invariant subspace.
This is the goal of the next subsection, where we shall make use of the fact  that $\Vv\Phi$ is the composition of a unitary operator and an orthogonal projection.
  
\subsection{Spectral properties of $\Vv\Phi$}

We analyse in this subsection the different contributions appearing in (\ref{dectvbnums}) from the spectral point of view. To do so, we endow $\Bb(\Ff_{-})$ with the Hilbert-Schmidt scalar product
\begin{equation}\label{hsscalprod}
  \langle X , Y  \rangle = \tr ( X^\ast Y ), \ \ \forall\, X,Y \in \Bb(\Ff_{-})
\end{equation}
and denote the adjoint of an operator $\Aa$ on $\Bb(\Ff_{-})$ with respect to (\ref{hsscalprod}) by $\Aa^\ast$.

\vspace{1mm}

We first consider the basic operators appearing in (\ref{dectvbnums}). 

\begin{lem}\label{structop}
Let  $\Vv, \Bb^{\mu\nu}$, and $\Phi $ be the operators on the Hilbert space $\Bb(\Ff_{-})$ defined  by (\ref{defvcal}), (\ref{defbmunu}), (\ref{eq-def_CPmapPhi}). Then \\
i)  $\Vv$ is unitary, \ie $\Vv^{-1}=\Vv^{\ast}$ \\
ii) $\{ \Bb^{\mu\nu}\}_{\mu, \nu \in \setmu}$ is a resolution of the identity by means of orthogonal projectors, \ie for all $\mu, \mu', \nu, \nu'  \in \setmu$ one has
\begin{align} \label{eq-orthonormal_family_projectors}
&\Bb^{\mu\nu}
=(\Bb^{\mu\nu})^\ast, \ \ \Bb^{\mu\nu}\Bb^{\mu'\nu'}=\delta_{\mu,\mu'}\delta_{\nu,\nu'}\Bb^{\mu\nu}, \ \ 
\sum_{\mu, \nu\in \setmu}\Bb^{\mu\nu}=\Identity_{\Bb(\Ff_-)}.
\end{align}
Consequently, $\Phi$ is also an orthogonal projector.
\end{lem}

\begin{proof} Point i) is well known. Point ii) follows from 
	{ the fact that the $B^\mu$ are the spectral projectors of a self-adjoint operator.
		}
\end{proof}

According to Lemma \ref{structop},  $\Vv \Phi$ is a contraction, $\|\Vv \Phi\| = \| \Phi\| \leq 1$. Since $\Bb (\Ff_-)$ is finite dimensional, $(\Vv \Phi)^n$ is readily computed from its spectral decomposition, see {\it e.g.} \cite{Ka}:
\begin{align}
\Vv \Phi=\sum_{k=1}^m \big( \lambda_k\Pp_k+\Dd_k\big)\;, 
\end{align}
where $\spec (\Vv \Phi)=\{\lambda_k\}_{1\leq k\leq m}$ is the spectrum of $\Vv \Phi$ and $\{\Pp_k\}_{1\leq k\leq m}$
(resp. $\{\Dd_k\}_{1\leq k\leq m}$) are the eigenprojectors (resp. eigennilpotents) of $\Vv\Phi$.  
Note that $\spec (\Vv \Phi)\subset \mathbb D$, where $\mathbb D=\{|z|\leq 1\}$.    
Hence, using $\Pp_k \Pp_l = \delta_{k,l} \Pp_k$ and $\Pp_k \Dd_l = \Dd_l \Pp_k = \delta_{k,l} \Dd_k$, one has
\begin{align}
(\Vv \Phi)^n=\sum_{k=1}^m\bigg( \lambda_k^n\Pp_k+\sum_{r=1}^q \Dd_k^r\lambda_k^{n-r}\begin{pmatrix} n \cr  r
\end{pmatrix} \bigg) , \ \ \mbox{$n\geq q$},
\end{align}
where $q$ is the maximal index of nilpotency of the $\Dd_k$'s. 
It follows from the contraction  property of $\Vv \Phi$, which implies $\|\Pp_k (\Vv \Phi)^n\|\leq 1$ for all $n\geq 0$,
that the   eigennilpotent and eigenprojector associated to an eigenvalue $\lambda_k$ of modulus one
are such that  $\Dd_k=0$ and  $\Pp_k$ is of norm one.
Thus for all eigenvalues $\lambda_k$ on the unit circle, $\Pp_k=\Pp_k^\ast$ is an orthogonal projector (see Theorem 2.1.9 in \cite{Si}).

We set
\begin{align}\label{specprojup}
&\Pp_{\bigcirc}=\sum_{k\  \mbox{\rm \tiny s.t.}\atop |\lambda_k|=1}\Pp_k, \ \ \Pp_{<}=\sum_{k\  \mbox{\rm \tiny s.t.} \atop  |\lambda_k|<1}\Pp_k, \ \ \mbox{so that}\nonumber \\
&\Pp_{\bigcirc}=\Pp_{\bigcirc}^\ast, \ \ \Pp_{<}=\Pp_{<}^\ast, \ \Pp_{\bigcirc}\Pp_{<}=\Pp_{<}\Pp_{\bigcirc}=0, \ \ \Pp_{\bigcirc}+\Pp_{<}=\Identity_{\Bb(\Ff_-)},
\end{align}
and decompose $\Bb (\Ff_-)$ in two orthogonal subspaces accordingly
\be
\Bb (\Ff_-)=\Pp_\bigcirc \Bb (\Ff_-)+\Pp_< \Bb (\Ff_-):=\Bb_\bigcirc+\Bb_<\;.
\ee
Then, $\exists \; C, \gamma>0$ such that 
\begin{align}\label{expdecvp}
&\Vv \Phi |_{\Bb_<} \ \ \mbox{satisfies} \ \ \|(\Vv \Phi |_{\Bb_<})^n\| \leq C\E^{-\gamma n}\\
&\Vv \Phi |_{\Bb_\bigcirc} \ \ \mbox{is unitary}, \nonumber
\end{align}
where any $\gamma<\min \{ |\ln(|\lambda_k|)|\,;\, \lambda_k \in \spec( \Vv \Phi) \setminus {\mathbb S}^1  \}$ will do.
The restriction $\Vv \Phi |_{\Bb_<}$ is completely non unitary, and the above provides in our setup the explicit decomposition of any contraction into a unitary restriction and a completely non unitary contraction, see \cite{SNF}.

To determine the eigenvalues of $\Vv\Phi$ on the unit circle ${\mathbb S}^1$, we use the following general { result} taken from \cite{HJ2}:
\begin{lem} \label{kerzero}
  Let $U$ be unitary and $P$ be an orthogonal projector  on a (possibly infinite dimensional) Hilbert space
  $\Hh$. For any $\varphi\in \Hh$ and $\theta \in \real$ one has
\begin{eqnarray}
\label{evpu}
  UP\varphi =e^{i\theta}\varphi  & \Leftrightarrow & 
  \varphi=P\varphi \ \mbox{and }\ e^{i\theta}\varphi=U\varphi=P UP  \varphi
  \\
\label{evpup}
& \Leftrightarrow &
P U\varphi =e^{i\theta}\varphi .
\end{eqnarray}
Moreover, writing $Q=\identity -P$, it holds
\begin{align}
  \ker  Q UP  =\{0\} \quad \Rightarrow \quad \spec_p(UP)\cap {\mathbb S}^1 =\spec_p(P U)\cap {\mathbb S}^1
  =\spec_p(P UP )\cap {\mathbb S}^1=\emptyset\;,
\end{align}
where $\spec_p$ denotes the point spectrum.
\end{lem}
According to Lemmas~\ref{structop} and \ref{kerzero},  if $e^{i \theta}$ is an eigenvalue of $\Vv \Phi$ on the unit circle
 then it is also an eigenvalue of $\Vv$ with the same eigenvector, and this eigenvector is invariant under $\Phi$.
Note that $\identity\in \Bb(\Ff_-)$ is an invariant vector under both $\Vv$ and $\Phi$, so that $1 \in \spec(\Vv \Phi) \cap  {\mathbb S}^1$.

\subsection{Estimate on the rest of the series for $\Ttt_t(X)$}

{ Let $\Delta >0$ be the minimal spectral gap of the coupling operator $T$,  
\begin{equation}
\Delta = \min_{\mu,\nu \in \spec (T), \mu \not= \nu} | \mu - \nu|\;.
\end{equation}	
}

\begin{thm} \label{theo-bounding_off_diagonal_contributions} 
	Let  $\gamma < \min \{ |\ln(|\lambda_k|)|\,;\, \lambda_k \in \spec( \Vv \Phi) \setminus {\mathbb S}^1  \}$. 
	Assume that the initial state of the reservoir is the quasi-free state satisfying \eqref{eq-initial_state_boson_action_on_Weyl} with $K \geq \identity$.	Then there exist $\lambda_0>0$ and $0<C_0<\infty$ such that $|\lambda|>\lambda_0$ implies that for all $t\in \mathbb N$ and $X \in \Bb( \Ff_-)$,
	\be
	\|\Ttt_t(X)-(\Vv\Phi)^t(X) \|\leq C_0\|\Pp_<(X)\|  t \E^{-\gamma t /2} { \E^{-\Delta \lambda^2/4}\;.}
	\ee
\end{thm}
\begin{proof}
	Recall that for $1\leq s \leq t$, the general term in the series  (\ref{dectvbnums})  defining $\Ttt_t(X)$ reads
	\begin{align}\label{genteser}
	& \sum_{1 \leq j_1 < j_2 < \cdots < j_s \leq t} \;\; \sum_{{\mu_{j_1}\not= \nu_{j_1}, \ldots, \mu_{j_s}\not= \nu_{j_s} }\atop {\nu_{j_i}, \mu_{j_i} \in \setmu}}
	(\Vv\Phi)^{j_{1}-1}\Vv\Bb^{\mu_{j_1}\nu_{j_1}}(\Vv\Phi)^{j_2-j_{1}-1}\Vv\Bb^{\mu_{j_2}\nu_{j_2}}\nonumber \\ 
	&\hspace{4mm}\cdots (\Vv\Phi)^{j_s-j_{s-1}-1}\Vv\Bb^{\mu_{j_s}\nu_{j_s}}(\Vv\Phi)^{t-j_s}(X)
	\exp \bigg\{ -\frac{\lambda^2}{4} \sum_{p,q=1}^s ( \mu_{j_p}-\nu_{j_p}) 
	( \mu_{j_q} - \nu_{j_q} )
	\big\langle \delta_{j_p}\, , K\, \delta_{j_q} \big\rangle \bigg\}.
	\end{align}
	Thus we are lead to consider the operator
	$\Vv\Bb^{\mu \nu}(\Vv\Phi)^n$ for $n\in \mathbb N$ and  $\mu\neq \nu \in \setmu$.
	We know from Lemma~\ref{kerzero} that $\range \Pp_{\bigcirc} \subset \range \Phi$, where $\Pp_{\bigcirc}$ is defined in \eqref{specprojup}. 
	Hence 
	\begin{equation} \label{eq-commutation_Pp_Phi}
	\Pp_{\bigcirc} = \Phi \,\Pp_{\bigcirc}= \Pp_{\bigcirc} \,\Phi \;.
	\end{equation}
	We now show that
	$\Vv \Pp_{\bigcirc}  = \Pp_{\bigcirc} \Vv$.
	Let $\lambda_k \in \spec(\Vv \Phi) \cap  {\mathbb S}^1$ and $\Pp_{k}$ be the spectral projector of $\Vv \Phi$ associated to $\lambda_k$. By Lemma~\ref{kerzero} again, 
	\begin{equation}
	\range \Pp_{k} \subset \ker \big( \Vv - \lambda_k \Identity_{\Bb(\Ff_-)} \big) 
	= \ker \big( \Vv^{-1} - {\lambda}_k^{-1} \Identity_{\Bb(\Ff_-)} \big)\;.
	\end{equation}
	Thus 
	\begin{equation} \label{eq-commutation_Pp_k_and_Vv}
	X \in \range  \Pp_k \;\;\Rightarrow\;\; 
	\Pp_k \Vv ( X) = \lambda_k \Pp_k (X) = \Vv \Pp_k (X)\;. 
	\end{equation}
	On the other hand, if $X \in (\range  \Pp_k)^\bot$ then for any $Y \in \Bb( \Ff_{-})$,
	\be
	\langle \Vv (X) , \Pp_k (Y)\rangle = \langle X , \Vv^{-1} \Pp_k (Y)\rangle =  {\lambda}_k^{-1} \langle X , \Pp_k (Y) \rangle=0.
	\ee
	This shows that $\Vv (X) \subset ( \range  \Pp_k)^\bot$. Therefore,
	\begin{equation} \label{eq-commutation_Pp_k_and_Vvbis}
	X \in ( \range  \Pp_k)^\bot \;\;\Rightarrow\;\; \Pp_k \Vv ( X) = \Vv \Pp_k ( X) = 0\;.
	\end{equation}  
	One infers from \eqref{eq-commutation_Pp_k_and_Vv} and \eqref{eq-commutation_Pp_k_and_Vvbis} that $\Pp_k \Vv = \Vv \Pp_k$, hence, as stated above, 
	\begin{equation} \label{eq-commutation_Pp_Vv}
	\Vv \Pp_{\bigcirc}  = \Pp_{\bigcirc} \Vv\;.
	\end{equation}

	For $\mu \not= \nu$, one deduces from \eqref{eq-commutation_Pp_Phi}, \eqref{eq-commutation_Pp_Vv} and $\Bb^{\mu \nu} \Phi = \Phi \Bb^{\mu \nu} = 0$ that  
	\be
	\Vv\Bb^{\mu \nu}\Pp_{\bigcirc}=\Pp_{\bigcirc}\Vv\Bb^{\mu \nu}=0 
	\ee
	and thus
	\begin{equation}
	\ \  [\Vv\Bb^{\mu \nu},\Pp_{<}]=0, \ \ \mbox{since} \ \ \Pp_{\bigcirc}+\Pp_{<}=\Identity_{\Bb(\Ff_-)}\;.
	\end{equation}
	Therefore,
	\be
	\Vv\Bb^{\mu \nu}(\Vv\Phi)^n=\Pp_<\Vv\Bb^{\mu \nu}\Pp_{<}(\Vv\Phi)^n\Pp_<, \ \ \forall \ \mu\neq \nu \in \setmu
	\ee
	so that, thanks to (\ref{expdecvp}) and Lemma \ref{structop},
	\be
	\|\Vv\Bb^{\mu \nu}(\Vv\Phi)^n\|\leq C\E^{-\gamma n}\;.
	\ee

	Hence, making use of $\| \Vv \| = \| \Bb^{\mu\nu}\| = 1$ and bounding 
	the exponential factor in \eqref{genteser} as in \eqref{decayexpterm} by
\begin{equation}
\exp \bigg\{ -\frac{\lambda^2}{4} \sum_{p=1}^s ( \mu_{j_p} - \nu_{j_p})^2 \bigg\}\;\;\leq \;\; { \E^{-s \Delta \lambda^2/4}}\;, 
\end{equation}
the term (\ref{genteser}) is bounded from above in norm by
	\begin{align}
	\sum_{1 \leq j_1 < j_2 < \cdots < j_s \leq t} &\;\; \sum_{{\mu_{j_1}\not= \nu_{j_1}, \ldots, \mu_{j_s}\not= \nu_{j_s} }\atop {\nu_{j_i}, \mu_{j_i} \in \setmu}} 
	C^{s+1}\E^{-\gamma(t-s)} { \E^{- s \Delta \lambda^2/4}} \|\Pp_<(X)\|\nonumber\\
	=\; &
	{
	\begin{pmatrix}t\\s\end{pmatrix} \ell^s C^{s+1}\E^{-s\Delta \lambda^2/4}\E^{-\gamma(t-s)} \|\Pp_<(X)\|\;,
}
	\end{align}
	{ where we have set $\ell=| \spec(T)| (  | \spec(T)|-1)$ with $|\spec (T)|$ the cardinal of $\spec (T)$.}
	Let us set { $\kappa_\lambda =\ell \,C\,\E^{-\Delta \lambda^2/4}\,\E^{\gamma}=O(\E^{-\Delta \lambda^2/4})$ as $\lambda \rightarrow \infty$}. We obtain
	\begin{align}\label{bosum}
	\| \Ttt_t(X)-(\Vv\Phi)^t (X) \|  \
	& \leq \
	C \E^{-\gamma t}\sum_{s=1}^t ( \kappa_\lambda)^s \begin{pmatrix}t\\s\end{pmatrix} \|\Pp_<(X)\|  \nonumber
	\\[1mm]
	& \leq  \ C \E^{-\gamma t} \, t \, \kappa_\lambda (1+\kappa_\lambda )^{t-1}  \|\Pp_<(X)\| \nonumber\\[1mm]
	&= \ C \frac{\kappa_\lambda}{1+\kappa_\lambda}\, t \, \E^{-t(\gamma -\ln(1+\kappa_\lambda))}  \|\Pp_<(X)\|\nonumber\\[1mm]
	& \leq C {  \kappa_\lambda}\, t \, \E^{-t(\gamma -\ln(1+\kappa_\lambda))}  \|\Pp_<(X)\|\;.
	\end{align}
	Suppose $\lambda^2$ is large enough so that $\ln(1+\kappa_\lambda)\leq \gamma/2$, which defines
	{ \begin{equation}
	\lambda_0^2= \Delta^{-1} \big( 2 \gamma+ 4 \ln (\ell C)- 4 \ln ( 1-\E^{-\gamma/2}) \big)\;.
	\end{equation}
}
	Then the exponential in \eqref{bosum} is bounded by $\E^{-\gamma t /2}$ 
	for all $t\in \mathbb N$. This yields the result.
\end{proof}

We consider now the estimate of the reminder $\Rr_t(X)$ in (\ref{leadingorder}), when we retain the leading order correction on the expression for $\Ttt_t(X)$.

\begin{cor}
	Under the asumptions of Theorem~\ref{theo-bounding_off_diagonal_contributions}, there exists $\lambda_0>0$ and $C_1<\infty$ such that $|\lambda|>\lambda_0$ implies that for all $t\in \mathbb N$ and $X \in \Bb( \Ff_-)$,
	\be
	\|\Rr_t(X)\|\leq C_1\|\Pp_<(X)\| t^2 \E^{-\gamma t/2} { \E^{-\Delta \lambda^2/2}}\;. 
	\ee
\end{cor}

\begin{proof}
	We can assume without loss of generally that $t \geq 2$. We simply have to consider (\ref{bosum}) with summation starting at $s=2$ instead of $s=1$. We thus get
	\begin{align}
	\| \Rr_t (X) \| 
	& \leq 
	C \E^{-\gamma t}\sum_{s=2}^t ( \kappa_\lambda)^s \begin{pmatrix}t\\s\end{pmatrix}  \|\Pp_<(X)\| 
	\leq 
	C \frac{\kappa^2_\lambda}{(1+\kappa_\lambda)^2} t(t-1) \E^{-t(\gamma -\ln(1+\kappa_\lambda))} \|\Pp_<(X)\| \;.
	\end{align}
	We conclude as in the proof of Theorem~\ref{theo-bounding_off_diagonal_contributions}. 
\end{proof}
\begin{rmk}
	Keeping all terms up to order $s$ in (\ref{dectvbnums}) only, we easily derive that the error term is of order { $\E^{-\Delta \lambda^2(s+1)/4}$}, multiplied by a factor decaying exponentially with time. The polynomial prefactors in $t$ are somehow immaterial since replacing $\gamma$ by $\gamma-\epsilon$, $\epsilon$ small enough, allows to dispense of them, at the cost of a larger constant $C$, see also (\ref{expdecvp}).
\end{rmk}

\section{Large Time Asymptotics }
\label{sec-convergence_large_time}

\subsection{Assumptions on the { coupling and fermion evolution operators.}}
\label{sec-asumption_V_and_T}

{ In this section we further specify the coupling operator $T\in {\cal B(F_-)}$ by assuming it is given by the second quantization $T=d\Gamma_-(\tau)$ of a self-adjoint operator $\tau$ on $\cal H$, where we assume hereafter that the single-particle Hilbert space has dimension $d \geq 2$. Explicitely, writing the spectral decomposition of $\tau\in {\cal B(H)}$ as
\be
\tau=\sum_{j=1}^d \varepsilon_j |f_j\rangle \langle f_j|
\ee
with eigenvalues $\varepsilon_j\in \mathbb R$ repeated according to multiplicity, and orthonormal eigenbasis $\{f_j\}_{j=1}^d$, we consider (see \eqref{eq-second_quantized_dGamma})
\be\label{secqtau}
T=\sum_{j=1}^d \varepsilon_j a(f_j)^*a(f_j).
\ee
According to the convention, we will write $a_j=a(f_j), a^*_j=a^*(f_j)$ and $n_j=a_j^\ast a_j$, $j=1,\ldots,d$. 

It follows from \eqref{eq-second_quantized_self_adjoint} that the eigenvectors of $T$ are of the form $\wedge^q f_{\jv}= f_{j_1} \wedge \cdots \wedge f_{j_q}$ for $q\in  \{0,\ldots, d\}$ and $\jv\in I_q$ an ordered multi-index, \ie $1\leq  j_1<j_2<\dots <j_q \leq d$, 
\be\label{evT}
 T \, (\wedge^q f_{\jv}) = \mu_{q,\jv}\, (\wedge^q f_{\jv})\;,\quad
\mu_{q,\jv} = \sum_{r=1}^q \varepsilon_{j_r}
\ee
with the  conventions $\wedge^0 f_{\jv}= \Omega$ and $ \mu_{0,\jv}=0$.
Since $\{ \wedge^q f_{\jv},\,    0\leq q\leq d,\,\jv\in I_q\}$ forms a basis of ${\cal B(F_-)}$, 
the spectrum of $T$ consists in $ \spec (T) = \{\mu_{q,\jv}, \,0\leq q\leq d,  \, \jv\in I_q\}$. 
We keep writing the spectral decomposition of $T\neq 0$ as
\be\label{stillspecdec}
T=\sum_{\mu \in \sigma(T)}\mu B^\mu,
\ee
where the eigenprojectors $B^\mu$ can be written explicitly in terms of products of the number operators $n_j$ and $\identity-n_{i}$. 
Note that $T$  leaves $\Hh^{\wedge n}$ invariant, thus the same is true for its spectral projectors $B_\mu$.

With $T$ of the form (\ref{secqtau}) we can complement Lemma \ref{structop} as follows:

\begin{lem}\label{structop_2}
	Let  $\Vv, \Bb^{\mu\nu}$, and $\Phi $ be the operators on the Hilbert space $\Bb(\Ff_{-})$ defined  by (\ref{defvcal}), (\ref{defbmunu}) and (\ref{eq-def_CPmapPhi}). Then,
	 the operators $\Vv,  \Phi$ and $ \Bb^{\mu\nu}$ leave
$ \Bb(\Hh^{\wedge n}, \Hh^{\wedge m})$ invariant for all $0\leq n, m \leq d$ and $\mu, \nu\in \setmu$.
\end{lem}

\begin{proof}  By construction, $\Gamma_{-} ( V)$ and $B^\mu$ preserve the number of particles,
\ie they leave  $\Hh^{\wedge n}$ invariant for all  $1\leq n\leq d$.
More precisely, recall that  $\{ \wedge^n f_\iv \}_{\iv \in I_n}$
is an orthonormal basis of $\Hh^{\wedge n}$. 
Thus  $\Bb(\Hh^{\wedge n}, \Hh^{\wedge m})$ is spanned by the rank one operators  
\begin{align}\label{raon}
	|\wedge^m f_{\jv}\rangle \langle \wedge^n f_{\iv}| \quad, \quad \iv \in I_n, \jv \in I_m\;.
\end{align}
Since 
\be
\Gamma_- (V^{-1})  \wedge^n f_{\iv} = V^{-1}f_{i_1}\wedge \dots \wedge V^{-1}f_{i_{n}} \ \in \ \Hh^{\wedge n}\quad, \quad \iv \in I_n\;,
\ee
see \eqref{+eq=definition_second_quantizationGamma}, one deduces from \eqref{defvcal} that 
$\Vv\Bb(\Hh^{\wedge n}, \Hh^{\wedge m})\subset \Bb(\Hh^{\wedge n}, \Hh^{\wedge m})$.
Similarly, since the projectors $B^\mu$ leave $\Hh^{\wedge n}$ invariant, it follows that
$\Bb(\Hh^{\wedge n}, \Hh^{\wedge m})$ is invariant under $\Bb^{\mu\nu}$ for any $\mu,\nu \in \setmu$. The same is true for $\Phi= \sum_{\mu\in \setmu} \Bb^{\mu\mu}$.
\end{proof}

}

We now look for the eigenvalues and eigenvectors of $\Vv$.
Recall that $\Vv(\cdot)=\Gamma(V^{-1})\cdot \Gamma(V)$, where $V\in \Bb(\Hh)$ is unitary. Let $\{\psi_i\}_{i=1}^d$ be an orthonormal basis of $\Hh$ of eigenvectors of $V$, with corresponding eigenvalues $\E^{i\alpha_i}$, $1\leq i \leq d$.
Then, using
a similar shorthand $ \wedge^n\psi_\iv = \psi_{i_1}\wedge \dots\wedge\psi_{i_n}$
as in (\ref{shortnotwed}) and  the convention  $\wedge^0\psi_\iv=\Omega$,
$\{ \wedge^n\psi_\iv \}_{\iv \in I_n}$  
is an orthonormal basis of $ \Hh^{\wedge n}$ for any $0\leq n\leq d$.
{ Moreover, by \eqref{+eq=definition_second_quantizationGamma},} 
\begin{align}
\Gamma_{-} (V) (\wedge^n \psi_{\iv}) =\E^{i(\alpha_{i_1}+\dots +\alpha_{i_n})} \wedge^n \psi_{\iv}\quad , \quad \iv \in I_n ,
\end{align}
{ (recall that $\Gamma(V)\Omega=\Omega$).}
Hence we have for all
$0\leq n,m\leq d$,
\begin{align}
  \Vv\big(|\wedge^n\psi_\jv\rangle\langle \wedge^m\psi_\iv|\big)=\E^{-i\big(\sum_{r=1}^n\alpha_{j_r}-\sum_{s=1}^m\alpha_{i_s}\big)}
  |\wedge^n\psi_\jv \rangle\langle \wedge^m\psi_\iv | 
\end{align}
and $\Bb(\Ff_-)$ is spanned by the eigenvectors of $\Vv$ given by the rank-one operators
$|\wedge^n\psi_\jv\rangle\langle \wedge^m\psi_\iv|$.
Note also that the eigenvalue $1\in \spec(\Vv)$ is at least $\dim (\Ff_-)$-fold degenerate, since 
\be
|\wedge^n\psi_\iv\rangle\langle \wedge^n\psi_\iv|
\ \subset \ \ker (\Vv-\Identity_{\Bb(\Ff_-)})
\ee
for all $\iv \in I_n$ and  $0\leq n\leq d$.
\medskip

We work under the following spectral non-degeneracy assumption on $\Vv$: \\

\noindent ({\bf SND}) {\it 
$\Vv$ has simple eigenvalues save for the eigenvalue $1$ which is $\dim (\Ff_{-})$-fold degenerate, i.e., 	
for all $0\leq n, m ,n',m'\leq d$, any  $(\iv, \jv)\in I_m \times I_n$ such that $\iv \neq \jv$ and any $(\iv', \jv') \in I_{m'} \times I_{n'}$, one has
\begin{align}
\Big( \sum_{r=1}^n\alpha_{j_r}-\sum_{s=1}^m\alpha_{i_s} \Big) - 
\Big( \sum_{r=1}^{n'}\alpha_{j_r'}-\sum_{s=1}^{m'}\alpha_{i_s'} \Big) 
\not\in 2\pi \integer \  \ \mbox{ when }\  (\iv, \jv) \not= (\iv',\jv')\;,
\end{align}
with the convention  $\sum_{r=1}^0\alpha_{j_r}=0$ for $\jv \in I_0$}.\\	

Consequently, if $\E^{\I\theta}\in \spec (\Vv)\setminus\{1\}$, this eigenvalue is simple and its eigenvectors are proportional to
\be \label{eq-eigenvectors_Vv}
|\wedge^n\psi_\jv \rangle\langle \wedge^{m}\psi_\iv |
\ee
for some $0\leq n,m \leq d$ and $\jv \in I_m$, $\iv \in I_n$
determined by $e^{\I\theta}$.
Now by Lemma~\ref{kerzero}, $e^{\I \theta} \in  {\mathbb S}^1 \cap \spec ( \Vv \Phi)$ if and only if  
$e^{\I \theta} \in  \spec ( \Vv )$ and the corresponding  eigenvectors are in the range of the projector $\Phi$. 
{ Moreover, the eigenvectors of $\Vv \Phi$ are also eigenvectors of $\Vv$ with the same eigenvalue $\E^{\I \theta}$, thus 
	they have the form \eqref{eq-eigenvectors_Vv} under assumption ({\bf SND}) when $\E^{\I \theta }\not= 1$.}
The following lemma gives a necessary and sufficient condition for the eigenvectors
\eqref{eq-eigenvectors_Vv} of $\Vv$ to be in the range of $\Phi$.

\begin{lem} \label{lem-NSC_range_of_Phi} Let $\phi_1, \phi_2 \in \Ff_-$, $\phi_1, \phi_2 \not= 0$.
  Then
  \begin{align} \label{eq-lemma_6.4}
|\phi_1\rangle\langle\phi_2|& =\Phi\big(|\phi_1\rangle\langle\phi_2|\big)
  \end{align}
 if and only if there exists a unique { $\tilde\mu\in \setmu$} such that 
\begin{align}  \label{eq-lemma_6.4_bis}
    B^{\tilde \mu} \phi_1=\phi_1 \ \ \mbox{and} \ \ B^{\tilde \mu} \phi_2=\phi_2.
\end{align}
\end{lem}

\begin{proof}
Assume that \eqref{eq-lemma_6.4} is true.  
Then, thanks to { the} definition \eqref{eq-def_CPmapPhi} of $\Phi$,
for all $(\mu, \nu) \in \setmu^2$,
\be \label{eq-proof_lemma6.4}
\mu \not= \nu \quad \Rightarrow \quad 
B^\mu  |\phi_1\rangle\langle \phi_2| B^\nu=0\; , \ \ \mbox{ i.e. } \ \ B^\mu\phi_1=0  \mbox{ or }  B^\nu\phi_2=0\;.
\ee
{ Since $(\sum_{\mu \in \setmu}B^\mu ) \phi_1 = \phi_1\not= 0$,
there exists some  $\tilde{\mu}\in \setmu$ such that  $B^{\tilde{\mu}} \phi_1 \not=0$.
 Then by taking  $(\tilde\mu,\nu)$ with $\nu\neq \tilde\mu$
  in \eqref{eq-proof_lemma6.4} one gets
  $B^\nu \phi_2 = 0$, implying that $B^{\tilde\mu} \phi_2 = \phi_2 \not= 0$. Taking now  $(\mu ,\tilde\mu)$ with $\mu\neq\tilde\mu$ in \eqref{eq-proof_lemma6.4}, one deduces that 
  $B^\mu \phi_1=0$, \ie $B^{\tilde\mu} \phi_1 = \phi_1$. Thus \eqref{eq-lemma_6.4_bis} is true.
The converse implication follows directly from $B^\mu B^\nu=\delta_{\mu \nu} B^\mu$ satisfied by the spectral projectors. }
\end{proof}

Relying on this lemma, we obtain
\begin{cor} \label{prop_eigenvectors_of_Vv_Phi_on_unit_circle}
    Assume that {\rm ({\bf SND})} holds. Then
$\E^{i\theta}\in \spec (\Vv)\cap  \spec (\Vv\Phi)\setminus\{1\}$ iff the corres\-ponding eigenvector
$|\wedge^n\psi_\jv \rangle\langle \wedge^{m}\psi_\iv |\in \Bb(\Ff_-)$ is such that
\be
\wedge^n\psi_\jv \in  { \range B^\mu }
\, , \quad
\wedge^{m}\psi_\iv \in { \range B^\mu }
\ee
for some unique $\mu\in \setmu$.
\end{cor}

{  According to~\eqref{evT}, the range of the spectral projectors $B^\mu$, $\mu\in\setmu$, is spanned by the orthonormal eigenvectors $\wedge^q f_{\jv}$ with eigenvalue $\mu_{q,\jv}=\mu$. More explicitely, if $\mu \not= 0$ then
\be\label{defiqmu}
\range B^\mu = 
 \text{span} \{\wedge^q f_{\jv}\,;  \  q=1,\ldots, d, \ \jv\in I_q^\mu \}
 \quad  , \quad   I_q^\mu = \big\{ \jv \in I_q \ ; \ \sum_{r=1}^{q} \varepsilon_{j_r} = \mu\big\}
\ee
and for $\mu=0$, $ \text{ran } B^0$ is spanned by $\Omega$ and the vectors $\wedge^q f_{\jv}$ with $q \geq 1$ and $\sum_{r=1}^q \varepsilon_{j_r} = 0$.
We will assume that the coupling operator $T$ satisfies a mild non-degeneracy assumption, stating that its restrictions to $\Hh^{\wedge n}$, for $0<n<d$, are not multiple of the identity:
\\

\noindent ({\bf MND}) {\it For all $0<n<d$, the restriction $T|_{\Hh^{\wedge n}}\neq c \,\identity|_{\Hh^{\wedge n}}$ for any $c\in \real$.}
}

\vspace{2mm}

Recalling that $T = \D \Gamma_{-} ( \tau)$, ({\bf MND}) is equivalent to 
$\tau \neq c' \,\identity_\Hh$ for any $c' \in \real$.
\medskip

\begin{rmk}\label{specvec}
  The vectors $\Omega$ and $F=f_1\wedge f_2\wedge  \dots \wedge f_{d}$, which span $\Hh^{\wedge 0}$ and $\Hh^{\wedge d}$ respectively,  are eigenvectors of $\Gamma_-(V)$ with respective eigenvalues $1$ and $\det V$ and of { $T$ with respective eigenvalues $0$ and $\mu_F =\sum_{i=1}^d \varepsilon_i$. If $\mu_F =0$, both these vectors belong to 
$\range ( B^0)$, 
recall (\ref{stillspecdec}), in which case} by \eqref{defvcal} and Lemma~\ref{lem-NSC_range_of_Phi},
\be
\Vv\Phi (| \Omega\rangle\langle F |)=( \det V) | \Omega\rangle\langle F |\quad , \quad \Vv\Phi (| F\rangle\langle \Omega |)=(\det V^\ast) | F\rangle\langle \Omega | \;
\ee
and $\det V$ and $\det V^\ast$ belong to $\spec(\Vv\Phi) \cap  {\mathbb S}^1$. Moreover, since $\det V= \E^{i\sum_{j=1}^n \alpha_j}$,
 these eigenvalues are simple and different from $1$ by the spectral non-degeneracy assumption  {\rm ({\bf SND})}.
 
 { On the other hand, if $\mu_F \not=0$ then $\det V, \det V^\ast \notin \spec(\Vv \Phi)$ since 
 	$\Phi ( | \Omega\rangle\langle F | ) =\Phi ( | F\rangle\langle \Omega | ) = 0$.} 
\end{rmk}

\medskip

We now show that under some additional  assumption on the { diagonal elements of the  spectral projectors of $T$ in the eigenbasis of $V$},
$\det V$ and $\det V^\ast$ are the only potential eigenvalues of $\Vv\Phi$ on the unit circle, together with $1$.\\

{ 
\noindent
({\bf Diag}) 
{\it For all $1\leq n\leq d-1$, all $\kv\in I_n$ and all $\mu \in \spec (T|_{\Hh^{\wedge n}})$, one has $\bra{\wedge^n \psi_{\kv}} B^\mu \wedge^n \psi_{\kv} \rangle \not= 0$}.\\

Note that the condition $\langle \wedge^n \psi_{\jv} | B^\mu \wedge^n \psi_{\jv} \rangle \not= 0$ 
in ({\bf Diag}) is equivalent to $\wedge^n \psi_{\jv} \notin \ker B^\mu$. In view of \eqref{defiqmu}, ({\bf Diag}) holds if and only if for all $0 \leq n \leq d-1$ and $\mu \in  \spec (T|_{\Hh^{\wedge n}})$, there is some $\jv \in I_n^\mu$ such that $\braket{\wedge^n \psi_\kv}{\wedge^n f_\jv}\not=0$ for all $\kv \in I_n$.
}


We provide an argument in Section \ref{legec} below that supports the genericity of that condition for random unitary matrices. 
{ We have}

\begin{prop} \label{spectrum_of_Vv_Phi_on_unit_circle}
 { Assume {\rm ({\bf SND})}, {\rm ({\bf MND})} and {\rm ({\bf Diag})}.Then, 
\be \spec (\Vv\Phi)\cap \mathbb S^1 =\left\{
\begin{matrix}
	\{1, \det V, \det V^\ast\} & \text{if } F\in \ker T, \\[2mm]
	\{1\} \phantom{xxxxxxxxxx}& \text{otherwise.}\ \ \
\end{matrix}
\right.
\ee
}
\end{prop}

\begin{proof}
{ In the first case,}  by Remark \ref{specvec}, we need to show that any $\E^{\I\theta}\in\spec(\Vv)\setminus\{1,\det V, \det V^\ast \}$ does not belong to
  $\sigma(\Vv\Phi)$.
  Thanks to ({\bf SND}), these eigenvalues  $\E^{\I\theta}$ are simple
  and are associated to eigenvectors of the form $|\wedge^n\psi_\kv \rangle\langle \wedge^{m}\psi_\lv |$ { for some} $(n,m)\not \in  \{(0,d), (d,0), (0,0), (d,d)\}$
  and $\kv\in I_n$, $\lv \in I_m$.
   Let us assume that  $\E^{\I\theta}\in\spec(\Vv \Phi)$. Then, by Corollary~\ref{prop_eigenvectors_of_Vv_Phi_on_unit_circle},
  { $\wedge^n \psi_\kv \in \range B^\mu$ and $ \wedge^m \psi_\lv \in \range B^\mu$}
  for some $\mu \in \sigma (T)$.
{ Given the above condition on $(n,m)$, one has either $0<  n <d$ or $0< m < d$. Let us first assume $0<  n <d$. Since $\wedge^n \psi_\kv = B^\mu(\wedge^n \psi_\kv)$, one has $\mu \in \sigma (T|_{\Hh^{\wedge n}})$ and for any $\nu \in \sigma( T|_{\Hh^{\wedge n}})$, $\nu \neq \mu$, $B^\nu (\wedge^n \psi_\kv) = B^\nu B^\mu (\wedge^n \psi_\kv) =0$. 	The existence of such $\nu \neq \mu$ is ensured by ({\bf MND}). This leads to a contradiction with ({\bf Diag}). A similar contradiction arises when $0<  m <d$, replacing $n$ by $m$ and $\wedge^n \psi_\kv$ by $\wedge^m \psi_\lv$. If $F\not\in \ker (T)$, the same argument holds for $\E^{\I\theta}\in\spec(\Vv)\setminus\{1\}$. 
}	
\end{proof}

\subsection{Invariant subspace of $\Vv \Phi$} \label{sec-invariant_subspace}

Let us consider now $\Pp_1$, the spectral projector associated to the eigenvalue $1$ of $\Vv\Phi$.
By Proposition~\ref{spectrum_of_Vv_Phi_on_unit_circle}, under assumptions ({\bf SND}), { {\rm ({\bf MND})} and ({\bf Diag})} one has, recall \eqref{specprojup},
{ 
\be \Pp_\bigcirc=\left\{
\begin{matrix}
	\Pp_1+\Pp_{\det V}+\Pp_{\det V^{\ast}} & \text{if $F\in \ker T$,}\\
	\Pp_1 \phantom{xxxxxxxxixxxxx}& \text{otherwise.}
\end{matrix}\right.
\ee
In the first case, }we already know that $\range \Pp_{\det V} = \complex \ketbra{\Omega}{F}$ and $\range \Pp_{\det V^\ast} = \complex \ketbra{F}{\Omega}$, see Remark~\ref{specvec}.

According to Lemma~\ref{kerzero} and the spectral non-degeneracy assumption ({\bf SND}), if $X \in \Pp_1 \Bb (\Ff_{-})$ then
\be
X \in \ker ( \Vv - \Identity_{\Bb(\Ff_-)} ) = \Span \big\{ \ketbra{\wedge^n \psi_\kv}{\wedge^n \psi_\kv} \big\}_{0 \leq n \leq d, \kv \in I_n}\;,
\ee
where we have set $\wedge^0\psi_\kv= \Omega$.
{ Moreover, Lemma \ref{structop_2}} shows that the range of $\Pp_1$ is the sum of the invariant subspaces of $\Vv \Phi$ in each $n$-particle sector, thus
\begin{align} \label{eq_RanP_1}
  \range \Pp_1 =\bigoplus_{n=0}^d \Pp_1 \Bb(\Hh^{\wedge n})
  \ \ \mbox{ with } \  \ \Pp_1\Bb(\Hh^{\wedge n}) \ \subset \ 
  \Span \{ |\wedge^n\psi_\kv \rangle\langle \wedge^{n}\psi_\kv |\}_{\kv\in I_n} \ \subset \  \Bb(\Hh^{\wedge n})\;.
\end{align}
We prove in this subsection that, under some additional assumptions on the unitary $V$,
$\Pp_1 \Bb(\Hh^{\wedge n})$ is reduced to $\complex \identity|_{\Hh^{\wedge n}}$. Hence the eigenvalue $1$ of $\Vv \Phi$ is $(d+1)$-fold degenerated.

{ By \eqref{eq_RanP_1}, } the elements  $X \in \Pp_1 \Bb (\Ff_{-})$ are of the form
\begin{align} \label{eq-spectral_decomp_X_basis_psi}
  X = \sum_{n=0}^d X_n \ \ \mbox{ with } \ \
  X_n= \sum_{\kv\in I_n} x_{\kv}^{(n)}  |\wedge^{n}\psi_\kv \rangle\langle \wedge^{n}\psi_\kv | \ \ \in \ \  \Bb ( \Hh^{\wedge n})\;,
 \end{align}
where $x_\kv^{(n)} \in \complex$.
One can in fact assume without loss of generality that $x_\kv^{(n)} \in \real$, \ie that $X_n$ is self-adjoint.
Actually, since  $B^\mu= (B^\mu)^\ast$, one has
$\Vv \Phi ( X^\ast_n) = (\Vv \Phi ( X_n) )^\ast$ for any $X_n \in \Bb (\Hh^{\wedge n})$. Hence $X_n \in \Pp_1 \Bb (\Hh^{\wedge n})$ if and only if
$\re(X_n) = (X_n + X^\ast_n )/2$ and $\im (X_n) = (X_n- X^\ast_n)/(2\I)  \in \Pp_1 \Bb (\Hh^{\wedge n})$.

\medskip

Let $X \in \Pp_1 \Bb (\Ff_{-})$ be self-adjoint. { We now show that }
\be \label{eq-invariant_states}
X = \sum_{n=0}^d x_n \,\identity_{\Hh^{\wedge n}} \quad , \quad x_{n} \in \real\;,
\ee
{ where $\identity_{\Hh^{\wedge n}}$ refers to the orthogonal projector onto $\Hh^{\wedge n}$},
under { the following hypothesis:} \\

\noindent
({\bf OffDiag}) {\it For any $1 \leq n \leq d-1$ and $\kv,\lv \in I_n$, $\kv\not= \lv$, one has $\bra{\wedge^n \psi_\kv} B^\mu \, {\wedge^n \psi_\lv}\rangle \not= 0$ for some $\mu \in \setmu$}.\\

\begin{rmk}\label{OffDiag-MND} { \textnormal{({\bf OffDiag})} implies \textnormal{({\bf MND})}. Indeed, let $0< n < d$. If 
	$T|_{\Hh^{\wedge n}} = c\,\identity |_{\Hh^{\wedge n}}$ for some $c\in {\mathbb R}$, then   
	$B^\mu |_{\Hh^{\wedge n}} = \identity |_{\Hh^{\wedge n}}$ or 
 $B^\mu |_{\Hh^{\wedge n}} =0$ for all $\mu \in \spec(T)$. This implies that for any 
 $\kv,\lv \in I_n$, $\kv\not= \lv$, one has 
 $\bra{\wedge^n \psi_\kv} B^\mu \, {\wedge^n \psi_\lv}\rangle=0$ for all $\mu$, in contradiction with \textnormal{({\bf OffDiag})}.
}
\end{rmk}

\vspace{1mm}

Let $1 \leq n < d$ and $X\in  \Pp_1 \Bb (\Ff_{-})$. To prove that \eqref{eq-invariant_states} holds under Assumption ({\bf OffDiag}), we first note that
by \eqref{eq-spectral_decomp_X_basis_psi},  the $x_\kv^{(n)}$ are eigenvalues of $X_n$ with eigenvectors $\wedge^n \psi_\kv$.
Assume that $X_n$ has two distinct eigenvalues $x_\kv^{(n)} \not= x_\lv^{(n)}$. Thanks to Lemma~\ref{kerzero} one has $\Phi ( X)=X$, thus
$X B^\mu = B^\mu X$ for any $\mu \in \setmu$. Since $B^\mu$ leaves $\Hh^{\wedge n}$ invariant,
$X_n B^\mu = B^\mu X_n$ holds as well. 
It follows that
$B^\mu (\wedge^n \psi_\kv)$ and  $B^\mu (\wedge^n \psi_\lv)$ are eigenvectors of $X$ with eigenvalues $x_\kv^{(n)}$ and $x_\lv^{(n)}$.
Because $X$ is self-adjoint, these two eigenvectors must be orthogonal, \ie
\be
\langle \wedge^n \psi_\kv | B^\mu \wedge^n \psi_\lv \rangle = 0\;,
\ee
for any $\mu \in \setmu$. This contradicts ({\bf OffDiag}). Hence $x_\kv^{(n)} = x_\lv^{(n)}$ for all $\kv,\lv \in I_n$, $\kv \not= \lv$,
implying that $X_n= x_n\,\identity|_{\Hh^{\wedge n}}$.
This equality holds for $n=0$ and $n=d$  as well since $\Hh^{\wedge 0}$ and { $\Hh^{\wedge d}$} are one-dimensional.

From the observation above, it follows that under assumption ({\bf OffDiag}),
all $X= \Pp_1 \Bb (\Ff_{-})$ are diagonal in each $n$-particle sector, being given by  \eqref{eq-invariant_states}
with complex constants $x_n \in \complex$.

\medskip

%
\subsection{Exponential Convergence to a Steady State}\label{sec-conv_equil}

We are now in a position to prove our main result.

\begin{thm} \label{theo-conv_equilibrium}
	Let the initial state $\omega_B$ of the reservoir be a quasi-free state satisfying (\ref{eq-initial_state_boson_action_on_Weyl}) with $K\geq 1$.  { Let the coupling operator $T$ { be} given by (\ref{secqtau}) and (\ref{stillspecdec}) and the free evolution on the sample $V$ be such that Assumptions {\rm ({\bf SND})}, {\rm ({\bf Diag})} and {\rm ({\bf OffDiag})} are satisfied. If $F\in \ker T$,} suppose { also that} the initial state of the sample
	$\omega_S(\cdot )= \tr ( \rho_S \, \cdot )$ is given by a density matrix $\rho_S$ such that 
	$\omega_S ( \ketbra{\Omega}{F}) = \omega_S ( \ketbra{F}{\Omega})=0$.
	 Then 
 there exists $\lambda_0>0$ such that for $|\lambda|>\lambda_0$ and for any $X \in \Bb(\Ff_-)$,
	\begin{equation} \label{eq-convergence_toward_equilibrium}
	\lim_{t \to \infty} \omega_S ( \Ttt_t ( X)) = \omega_S^\infty ( X)
	\end{equation}
	where the asymptotic state $\omega_S^\infty (\cdot )= \tr ( \rho_S^\infty \,  \cdot )$ has a density matrix given by
	\begin{equation}
	\rho_S^\infty = \bigoplus_{n=0}^d \begin{pmatrix} d \\ n \end{pmatrix}^{-1}\tr_{\Hh^{\wedge n}} ( \rho_S |_{\Hh^{\wedge n}})  \,\identity_{\Hh^{\wedge n}}\;.
	\end{equation}
	The convergence in \eqref{eq-convergence_toward_equilibrium} is exponential. 
\end{thm}
\begin{rmk} 
i) In case $\omega_{S}$ leaves all $n$-particle subspaces $\Hh^{\wedge n}$ invariant, \ie, 
\begin{align}
	\rho_S &= \bigoplus_{n=0}^d \rho_{S,n} \ \ \text{with} \ \ \rho_{S,n}\in \Bb(\Hh^{\wedge n}), \text{ then }\nonumber\\
	\rho_S^\infty &= \bigoplus_{n=0}^d \, \rho_{S,n}^\infty \ \ \text{with} \ \ \rho_{S,n}^\infty=  \begin{pmatrix} d \\ n \end{pmatrix}^{-1}\tr_{\Hh^{\wedge n}}(\rho_{S,n}) \,\identity_{\Hh^{\wedge n}}\in \Bb(\Hh^{\wedge n}).
\end{align}  \\
ii) The asymptotic state depends on the initial condition, so that, strictly speaking, it does not qualify to be an equilibrium state. However, as point i) shows, within each of the  $n$-particle subspaces $\Hh^{\wedge n}$ the asymptotic state coincides with an infinite-temperature Gibbs state: 
\begin{equation}
	\rho_S=\rho_{S,n} \ \Rightarrow \ \rho_{S}^\infty=\begin{pmatrix} d \\ n \end{pmatrix}^{-1}\identity_{\Hh^{\wedge n}}\; .
\end{equation}
\end{rmk}

\begin{proof}
{ We consider the case $F\in \ker T$, the other one being simpler.}  By Theorem \ref{theo-bounding_off_diagonal_contributions} one has
 \begin{equation}
 \omega_S \big( \Ttt_t (X) \big) = \omega_S \big( (\Vv \Phi)^t (X) \big) + 
 \omega_S \big( \Ttt_t (X) - (\Vv \Phi)^t (X) \big)
 \; \text{ with } \;  \omega_S \big( \Ttt_t (X) - (\Vv \Phi)^t (X) \big) \to 0
\end{equation}
as $t \to \infty$, the convergence being exponential. { By  Proposition~\ref{spectrum_of_Vv_Phi_on_unit_circle} and Remark \ref{OffDiag-MND}, one has
} 
\be
(\Vv \Phi)^t (X) = (\Vv \Phi)^t \Pp_{<} (X) + (\Vv \Phi)^t \Pp_{1} (X) + (\Vv \Phi)^t \Pp_{\det V} (X) 
+ (\Vv \Phi)^t \Pp_{\det V^\ast} (X) \;. 
\ee
{ Keeping in mind that $\range \Pp_{\det V} = \complex \ketbra{\Omega}{F}$ and $\range \Pp_{\det V^\ast} = \complex \ketbra{F}{\Omega}$ (see Remark~\ref{specvec}) and using} the assumption
$\omega_S ( \ketbra{\Omega}{F}) = \omega_S ( \ketbra{F}{\Omega})=0$, we get
\begin{equation}
\omega_S \big( (\Vv \Phi)^t (X)  \big) = \omega_S \big( (\Vv \Phi)^t \Pp_{<}  (X) \big)+ 
\omega_S \big( \Pp_{1}  (X) \big)
\end{equation}
for all $t \in \pinteger$.
By \eqref{expdecvp} the first term in the \RHS can be bounded by a constant times 
$\E^{-\gamma t}  \| \Pp_{<} (X)\|$ and thus vanishes exponentially fast as $t\to \infty$.
Now from \eqref{eq-invariant_states} one obtains 
\begin{equation}
\Pp_1 (X) =\sum_{n=0}^d \begin{pmatrix} d \\ n \end{pmatrix}^{-1} 
\tr_{\Hh^{\wedge n}} ( X|_{\Hh^{\wedge n}}) \,\identity_{\Hh^{\wedge n}}\;.
\end{equation}
The result follows from the self-adjointness of $\Pp_1$ so that $\omega_{S}(\Pp_1 (X))=\tr (\Pp_1(\rho_S)  X)$, 
and thus (\ref{eq-convergence_toward_equilibrium}) holds with 
$\rho_S^\infty=\Pp_1(\rho_S)$.
\end{proof} 


{ We get the following direct consequence when there is initially only one fermion in the sample and the observable $X$ conserves the number of particles:

\begin{cor}\label{cor5.9}
 Under the hypotheses of Theorem~\ref{theo-conv_equilibrium} on $\omega_B$ and $T$, assume that $\rho_S = \rho_{S,1}$, $X = \bigoplus_{n=0}^d \,X_n$ with $X_n \in \Bb ( \Hh^{\wedge n})$ and that
\begin{itemize}
\item[(a)] the map ${\cal V}_1$ defined by ${\cal V}_1 (X) = V^{-1} X V$, $X\in \cal{B} ( \cal{H} )$, has simple spectrum { save for the eigenvalue $1$ which is $d$-fold degenerate;};
\item[(b)] For all $k=1,\ldots,d$ and $\mu \in \spec ( \tau)$, $\langle \psi_k | B^\mu \psi_k \rangle \neq 0$.
\end{itemize}
If $\tau$ has at least one simple eigenvalue, then
%
\be
\lim_{t\rightarrow \infty} \omega_S({\mathcal T}_t(X)) =
\frac{1}{d} \tr_{\cal H}(X)\;.
\ee
\end{cor}

\begin{rmk} Hypotheses (a) and (b) correspond to the restriction of Assumptions {\rm ({\bf SND})} et {\rm ({\bf Diag})} to the one-particle subspace $\cal{H}$.
\end{rmk}
}
	
\begin{proof}
	{ 
	Since $\rho_S = \rho_{S,1}$ and $X = \bigoplus_{n=0}^d \,X_n$ with $X_n \in \Bb ( \Hh^{\wedge n})$, we have  by Lemma~\ref{structop_2}
	\begin{equation}  
  \omega_S \big( (\Vv \Phi)^t (X) \big)
   = \tr_\Hh \big( \rho_{S,1} (\Vv \Phi)^t (X) \big|_{\Hh}\big)
  = \omega_S  \big( (\Vv \Phi|_{\Bb ( \Hh)})^t (X_1)\big)\;.
  \end{equation}
 Observe that the restriction of  $\Vv \Phi$ to $\Bb ( \Hh)$ is given by replacing $\Gamma(V)$ by $V$ in \eqref{defvcal} and $B^\mu$ by the spectral projectors of $\tau$, which are given by $B^\mu|_\Hh$ with $\mu \in \spec (\tau)$.
 Doing these substitutions in the proofs of Proposition~\ref{spectrum_of_Vv_Phi_on_unit_circle} and Theorem~\ref{theo-conv_equilibrium}, one sees that assumptions
  ({\bf SND}) and ({\bf Diag}) can be replaced by (a) and (b). 
  Since $\tau$ has a simple eigenvalue $\varepsilon_i$, $B^\mu|_\Hh = \ketbra{f_i}{f_i}$ for $\mu = \varepsilon_i$. Thus 
		$|\langle\psi_k | B^\mu \psi_l \rangle | = \| B^\mu\psi_k\| \| B^\mu \psi_l\| \not=0$ by assumption (b).
Hence assumption {\rm ({\bf OffDiag})} is satisfied for $n=1$.  
One can then use the arguments at the end of Sec.~\ref{sec-invariant_subspace} to show that
$X_1 = \Pp_1 (X_1)$ $ \Rightarrow X_1 = x_1 \identity|_\Hh$ with $x_1 \in \complex$. 
	}
\end{proof}

\subsection{Model on a finite graph with the coupling operator $T_{\rm{hop}}$}
\label{sec-results_on_specific_model}

{ 
In this subsection, we analyse the specific model of fermionic QWs on 
a finite graph $\Lambda=\{1,\ldots, d\}$ discussed at the end of Sec.~\ref{sec-notation}, for which the operator $T$ is given by
\begin{equation} \label{eq-T_transport-model+bis}
T_{\rm{hop}}= \D \Gamma_{-}( \tau_{\rm{hop}}) \quad, \quad \tau_{\rm{hop}} = \E^{\I \varphi} \ketbra{e_2}{e_1} + \E^{-\I \varphi} \ketbra{e_1}{e_2}\;, 
\end{equation}
where $\{ e_l\}_{l=1}^d$ is the canonical basis of the one-particle states localized at site $l \in \Lambda$. Hereafter, we assume $d \geq 3$. The eigenvalues of $\tau_{\rm{hop}}$ are $\varepsilon_1 = 1$, $\varepsilon_2 = -1$ and $\varepsilon_3=0$, the first two eigenvalues being simple and the third one being $(d-2)$-fold degenerate.
A corresponding orthonormal basis of eigenvectors of $\tau_{\rm{hop}}$ is formed by the vectors 
}
\begin{align}\label{newcan}
f_1 &=\frac{1}{\sqrt{2}} \big( e_1+\E^{i\varphi}e_2 \big)\nonumber\\[2mm]
f_2 &=\frac{1}{\sqrt{2}} \big( e_1- \E^{i\varphi} e_2 \big) \\[2mm]
f_j & = e_j\quad, \quad j =2,\ldots, d\;. \nonumber
\end{align}
{ We denote by $a_+ = a ( f_{1} )$, $a_+^\ast = a^\ast ( f_{1} )$ and $n_+ = a_+^\ast a_+$  (respectively $a_- = a ( f_{2} )$, $a_{-}^\ast = a^\ast ( f_{2} )$  and $n_- = a_{-}^\ast a_{-}$) the annihilation, creation and number operators 
associated to the eigenvector $f_1$ (resp. $f_2$). Similarly,  
$a_j = a(e_j)$, $a_j^\ast = a^\ast ( e_j )$ and $n_j=a_j^\ast a_j$ are the  
annihilation, creation and number operators at site $j\in \Lambda$. Note
that our convention for $a^\sharp_1$ and $a^\sharp_2$ in this subsection is different from 
that adopted in Subsection~\ref{sec-asumption_V_and_T}. 

According to \eqref{evT}, the spectrum of $T_{\rm{hop}}$ is
$\spec (T_{\rm{hop}}) = \{ 1, -1, 0\}$. Let us denote by 
$B^+$, $B^{-}$ and $B^0$ the spectral projectors 
of $T_{\rm{hop}}$ associated to the eigenvalues $\mu_+=+1$, $\mu_{-}=-1$ and $\mu_0 = 0$.
The eigenspaces of $T_{\rm{hop}}$ are given \eqref{defiqmu} with
}
\begin{equation}
I_1^{+} = \{ 1\} \quad , \quad I_1^{-} = \{ 2\}\quad, \quad I_1^0 = \{ j\}_{j=3}^d 
\end{equation}
and, for $n \geq 2$,
\begin{align} \label{defiqmu_hop}
\nonumber  
I_n^+ & = \big\{ ( 1,j_2,\dots,j_n) \big\}_{2<j_2<\dots < j_n \leq d}
\\ 
I_n^- & = \big\{ ( 2,j_2,\dots,j_n)  \big\}_{2<j_2<\dots < j_n \leq d}\;.
\\ \nonumber
I_n^0 & = \big\{ ( 1,2,j_3,\dots,j_n)  \big\}_{2<j_3<\dots < j_n \leq d} \cup \big\{ (j_1,j_2,\dots, j_n)  \big\}_{2< j_1 < \dots < j_n \leq d }\;.
\end{align}
More explicitely, 
\begin{align}\label{carkerb}
\range B^+ &= \Span \ \{f_1\}\cup\{f_1\wedge f_{j_2}\wedge  \dots \wedge f_{j_n}\}_{ j_2>2\atop 2\leq n\leq d-1} \nonumber\\
\range B^{-}&=\Span \ \{f_2\}\cup\{f_2\wedge f_{j_2}\wedge  \dots \wedge f_{j_n}\}_{ j_2>2\atop2\leq n\leq d-1} \\
\range B^0 & =\Span \ \{\Omega\}\cup \{ f_j\}_{j>2} \cup \{f_1\wedge f_2\wedge f_{j_3}\wedge \dots \wedge f_{j_n}\}_{j_3>2 \atop 2\leq n\leq d}
\cup\{f_{j_1}\wedge  \cdots \wedge f_{j_n}\}_{ j_1>2\atop1\leq n \leq d-2} \nonumber
\;.
\end{align}
{ Thus, the spectral projectors $B^\mu$ are given by
\begin{equation} \label{eq-spctral_projector_transport_model}
B^\pm = n_{\pm} ( 1 - n_{\mp})\quad , \quad B^0 = \identity - B^+ - B^- = n_{+} n_{-} + (1-n_{+})(1- n_{-})\;. 
\end{equation}
Indeed, one checks that
\begin{align}
 B^+ - B^{-} & = n_{+} - n_{-} =  \frac{1}{2} \big( a_1^\ast + \E^{\I \varphi} a_2^\ast \big) \big( a_1 + \E^{-\I \varphi} a_2 \big) 
 - \frac{1}{2} \big( a_1^\ast - \E^{\I \varphi} a_2^\ast \big) \big( a_1 - \E^{-\I \varphi} a_2 \big) \nonumber \\
&  = \E^{\I \varphi} a_2^\ast a_1 + \E^{-\I \varphi} a_1^\ast a_2  =  T_{\rm{hop}} \;.
\end{align}
Note that the restriction of $T_{\rm{hop}}$ to $\Hh^{\wedge n}$ satisfy ({\bf MND}) and that $F\in \range B^0=\ker T_{\rm{hop}}$.
}
Now, consider the matrix 
\be\label{matcoef}
C = \begin{pmatrix}
	c_{1}(1) & c_{1}(2) & \cdots & c_{1}(d) \\
	c_{2}(1) & c_{2}(2) & \cdots & c_{2}(d) \\
	\vdots   & \vdots   & \ddots & \vdots   \\
	c_{d}(1) & c_{d}(2) & \cdots & c_{d}(d)
\end{pmatrix}
\in M_d(\mathbb C),
\ee
where $c_j(k)= \braket{f_j}{\psi_k}$ are the coefficients of the eigenvectors $\psi_k$ of $V$ in the orthonormal basis $\{f_i\}_{i=1}^d$,
\be
\psi_k=\sum_{j=1}^d c_j(k) f_j\quad , \quad k=1,\dots,d\;.
\ee
Thus $C$ is  the unitary matrix implementing the change from the { $\{ f_j\}$-basis} to the orthonormal basis of eigenvectors of $V$. 
Consequently, for all $0\leq n\leq d$ and $\kv\in I_n$,
\begin{align}\label{expwedgevec}
\wedge^n\psi_\kv =\sum_{\jv \in I_n}c_\jv (\kv) \wedge^n f_\jv \;,
\end{align}
where, thanks to (\ref{scalprodet}), $c_\jv (\kv )$ is the following $n\times n$ minor of the matrix $C$: 
\be
c_\jv (\kv ) =\big\langle \wedge^n f_\jv |  \wedge^n \psi_\kv \big\rangle = 
\begin{vmatrix} c_{j_1}(k_1) &c_{j_1}(k_2)& \cdots &c_{j_1}(k_n) \\
	c_{j_2}(k_1) &c_{j_2}(k_2)& \cdots &c_{j_2}(k_n)\\
	\vdots & \vdots &\ddots& \vdots\\
	c_{j_n}(k_1) &c_{j_n}(k_2)& \cdots &c_{j_n}(k_n)\end{vmatrix}\; . 
\ee
{ Assumptions ({\bf{Diag}}) and ({\bf{OffDiag})} can be rewritten in terms of the minors $c_\jv ( \kv)$ as follows:\\

\noindent
({\bf{Diag}}) {\it for $T=T_{\rm{hop}}$: 
For any $1\leq n \leq d-1$ and any $\kv \in I_n$, there exists 
$\jv^+ \in I_n^+$, $\jv^{-} \in I_n^{-}$ and $\jv^0 \in I_n^0$ such that 
$c_{\jv^{+}} ( \kv)\neq 0$, $c_{\jv^{-}} ( \kv)\neq 0$, and $c_{\jv^{0}} ( \kv)\neq 0$}. \\

\noindent
({\bf{OffDiag}}) {\it for $T=T_{\rm{hop}}$: For any $1 \leq n \leq d-1$ and $\kv,\lv \in I_n, \kv \not= \lv$, one has
\be
 \sum_{\jv \in I_n^\mu} \overline{c_\jv ( \kv )} c_\jv (\lv) \not= 0\quad \text{ for some } \mu \in \setmu\;.
\ee
}

In particular, the condition ({\bf{Diag}}) is satisfied  in the one-particle subspace when all eigenvectors
$\psi_k$ of $V$ are such that $\braket{f_j}{\psi_k} \not=0$ for $j=1$, $j=2$ and at least one $j>2$. 


\medskip \medskip

We now prove the convergence to a steady state for the fermionic QWs on 
a finite graph $\Lambda$ with coupling operator $T_{\rm{hop}}$ under an alternative set of hypotheses.

Let us denote for $n\in \{1,\ldots ,d\}$, $\Gamma_n(V) = \Gamma (V)|_{\Hh^{\wedge n}}= V \otimes \dots \otimes V$ and
 $B^\mu_n = B^\mu|_{\Hh^{\wedge n}}$, where $\mu \in \{ +,-,0\}$.\\

	\noindent \rm{({\bf Cyc})} {\it For any $2 \leq n \leq d-1$, there exists  a triplet $(\mu,\nu,\nu')$ with $\{ \mu, \nu, \nu'\}= \{ +.-,0\}$ such that
the $\ast$-algebra generated by the operators
\be 
B_n^\mu \Gamma_n (V) B_n^\mu \quad , \quad B_n^\mu \Gamma_n (V) B_n^{\nu} \Gamma_n (V) B_n^\mu
\quad \text{and} \quad B_n^\mu \Gamma_n (V) B_n^{\nu'} \Gamma_n (V) B_n^\mu
\ee
is the whole algebra $\Bb ( B_n^\mu \Hh^{\wedge n})$}.

\vspace{1mm}

\begin{prop} \label{theo-conv_equilibrium_hop}
	Let $T=T_{\rm{hop}}$ be given by \eqref{eq-T_transport-model+bis}.
	Let the initial state $\omega_B$ of the reservoir be a quasi-free state satisfying \eqref{eq-initial_state_boson_action_on_Weyl} with $K\geq 1$  and  the initial state of the sample
be such that $\omega_S ( \ketbra{\Omega}{F}) = \omega_S ( \ketbra{F}{\Omega})=0$. 
	Then the conclusion of Theorem~\ref{theo-conv_equilibrium} on the convergence  of the system state at large time holds under assumptions \rm{({\bf SND})}, \rm{({\bf Diag})} and 
	\rm{({\bf Cyc})}.
 \end{prop}

}

\begin{proof}
{ 
Since \rm{({\bf MND})} is satisfied for $T = T_{\rm{hop}}$ and $F \in \ker (T_{\rm{hop}})$, the part of the spectrum of $\Vv \Phi$ on the unit circle is given by $\{ 1, \det V, \det V^\ast\}$ by Proposition~\ref{spectrum_of_Vv_Phi_on_unit_circle}. 
Thus we only have to prove that under the assumptions of the Proposition, the invariant subspace $\Pp_1  \Bb (\Ff_{-})$ of $\Vv \Phi$ is reduced to observables of the form \eqref{eq-invariant_states}. From the argument after \eqref{eq-spectral_decomp_X_basis_psi}, one can without loss of generality assume that the observables invariant under $\Vv \Phi$ are self-adjoint.
}

Let $X \in \Pp_1 \Bb (\Ff_{-})$ { be self-adjoint} and  $1 \leq n < d$. Then $X$ is given by \eqref{eq-spectral_decomp_X_basis_psi}. Moreover, since $\Phi(X) = X$, one has
{ $B^\mu_n X_n B^\nu_n = \delta_{\mu,\nu} X_n B^\mu_n$.}
We make the following observation. Suppose that one can show that  
{ 
\be \label{eq-diagonal_projected_X}
B^\mu_n X_n B^\mu_n = X_n B^\mu_n = x_n^\mu B^\mu_n \quad \text{ for some $\mu \in \{ +,-,0\}$ and $x_n^\mu \in \real$.}
\ee
} 
Then one can use assumption { ({\bf Diag})} to prove that 
$X_n = x_n \,\identity|_{\Hh^{\wedge n}}$.
Actually, { \eqref{eq-spectral_decomp_X_basis_psi} and \eqref{eq-diagonal_projected_X} imply that  $\range B_n^\mu   \subset  \Span \{ \wedge^n \psi_\kv \,;\, x_\kv^{(n)} = x_n^\mu\}$.}
If there exists some $x_\jv^{(n)} \not= x_n^\mu$ in the decomposition \eqref{eq-spectral_decomp_X_basis_psi} then,
{ by orthogonality of the eigenspaces of $X_n$ (recall that $X_n$ is self-adjoint) and 
	$\range B^\mu_n \,\bot\,\ker B^\mu_n$, it follows that
 $\wedge^n \psi_\jv \in \ker B_n^\mu$, in contradiction with ({\bf Diag}). Thus  
{ ({\bf Diag})} and \eqref{eq-diagonal_projected_X}
}
imply that $X_n$ is proportional to the identity. 
Note that for $n=1$ the projectors $B^\pm |_{\Hh}$ are of rank one, so that  
{ \eqref{eq-diagonal_projected_X} is always satisfied for $\mu = \pm$ and 
	$x_1^\pm = \tr (X_1 B_1^\pm)$.}
The situation is more involved for $n \geq 2$.

{ Suppose that $n\geq 2$.}
In order to prove { \eqref{eq-diagonal_projected_X}}
we first work in the direct sum of the ranges of $B^+_n$ and  $B^{-}_n$ and introduce  
the projectors
\be
\Pi_n = B^{+}_n+ B^{-}_n 
\quad , \quad
\Pi_n^\bot = B^0_n \;.
\ee
Thus $X_n = \Pi_n X_n \Pi_n+ \Pi_n^\bot X_n \Pi_n^\bot$.
We now express the condition $\Vv(X)=X$, which means that $X$ commutes with $\Gamma (V)$.
By projecting the commutator $[X,\Gamma(V)]$, one is led to
\be \label{eq-commutator_condition}
\big[ \Pi_n X_n \Pi_n \, , \, \Pi_n \Gamma_n ( V) \Pi_n \big] = 0\;.
\ee
Let us write the matrices of the two operators of this commutator in the
basis $\{ \wedge^n f_\jv \}_{\jv \in I_n^+} \cup \{ \wedge^n f_\jv \}_{\jv \in I_n^{-}}$ of $\Pi_n \Hh^{\wedge n}$:
\be
{\rm Mat} (\Pi_n X_n \Pi_n )
= \begin{pmatrix}
	X_n^+ & 0 \\
	0    &   X_n^{-}
\end{pmatrix}
\, , \quad
{\rm Mat} (\Pi_n \Gamma_n (V) \Pi_n )
= \begin{pmatrix}
	A_{++}^n & A_{+-}^n \\
	A_{-+}^n    &  A_{--}^n
\end{pmatrix}
\ee
with $(A_{pq})_{\iv,\jv} = \bra{f_{(3- {\rm sign} (p))/2} \wedge f_{i_2} \wedge \dots \wedge f_{i_n}} \Gamma(V) \, f_{(3- {\rm sign} (q))/2} \wedge f_{j_2} \wedge \dots \wedge f_{j_n}\rangle$ for
$p,q \in \{+,-\}$.
The condition \eqref{eq-commutator_condition} can be rewritten as
\be \label{eq-commutation_cond1}
\left\{
\begin{array}{lcl}
	\big[ X_n^+ \, , \, A_{++}^n \big] & = & 0
	\\[1ex]
	\big[ X_n^- \, , \, A_{--}^n \big]
	& = & 0 
	\\[1ex]
	A_{+-}^n X_n^- - X_n^+ A_{+-}^n & = & 0
	\\[1ex]
	A_{-+}^n X_n^+ - X_n^- A_{-+}^n & = & 0\;.
\end{array}
\right.
\ee
Left-multiplying the third (respectively fourth) equation by $A_{-+}^n$ (resp. $A_{+-}^n$) and right-multiplying the fourth (resp. third)
equation by $A_{+-}^n$ (resp. $A_{-+}^n$), this gives
\be \label{eq-commutation_cond2}
\left\{
\begin{array}{lcl}
	\big[ X_n^+ \, , \, A_{+ -}^n A_{- +}^n \big] & =  & 0
	\\[1ex]
	\big[ X_n^- \, , \, A_{- +}^n  A_{+ -}^n \big] & = & 0\;.
\end{array}
\right.
\ee
The products $A_{+-}^n  A_{-+}^n$ and $A_{-+}^n  A_{+-}^n $ appearing in \eqref{eq-commutation_cond2}
are the matrices of the operators
\be
B^{+}_n \Gamma_n (V) B^{-}_n \Gamma_n (V) B^{+}_n
\quad \text{ and } \quad 
B^{-}_n \Gamma_n (V) B^{+}_n \Gamma_n (V) B^{-}_n \,.
\ee
Similarly, $A_{\pm \pm}^n$ is the matrix of the operator $ B^{\pm}_n \Gamma_n (V) B^{\pm}_n$. Thus, by (\ref{eq-commutation_cond1}) { and \eqref{eq-commutation_cond2}, $X^\pm_n$ commutes with both ${\rm Mat}(B^{\pm}_n \Gamma_n (V) B^{\pm}_n)$ and ${\rm Mat}(B^{\pm}_n \Gamma_n (V) B^{\mp}_n \Gamma_n (V)B^{\pm}_n)$.}

By repeating the argument above in the direct sum of the ranges of $B^{\pm}_n$ and  $B^{0}_n$ and
considering the basis $\{ \wedge^n f_\jv \}_{\jv \in I_n^\pm} \cup \{ \wedge^n f_\jv \}_{\jv \in I_n^{0}}$
of this subspace,
we obtain similarly that 
{ $X_n^0= {\rm Mat}( B_n^0 X_n B_n^0)$ commutes with  ${\rm Mat}( B_n^0 \Gamma_n (V) B_n^0)$ and that 
$X_n^\pm$ and $X_n^0$}
commute respectively with the matrices of the operators
\begin{equation}
	B^{\pm}_n \Gamma_n (V) B^{0}_n \Gamma_n (V) B^{\pm}_n
	\quad \text{ and } \quad
	B^{0}_n \Gamma_n (V) B^{\pm}_n \Gamma_n (V) B^{0}_n\;.
\end{equation}
From the above commutation relations, $X_n^\mu$ commutes with the matrices of
$B_n^\mu \Gamma_n (V) B_n^\mu$ and $B_n^\mu \Gamma_n(V) B_n^\nu \Gamma_n (V) B_n^\mu$ for $\nu \not= \mu$. 
{ Since it is self-adjoint, it also commutes with the adjoints of these operators. Hence $X_n^\mu$ commutes
with the $\ast$-algebra} generated by these operators.
Thus if ({\bf Cyc}) is satisfied 
then $X_n^\mu = x_n^\mu B_n^\mu$ for some { $x_n^\mu \in \real$, with $X_n^\mu = B_n^\mu X_n B_n^\mu$.} 
By the observation above one concludes that { $x_n^\mu$} is independent of $\mu$ and  $X_n = x_n \identity_{\Hh^{\wedge n}}$.
Thus, Assumption ({\bf Cyc}) implies that  $X \in \Pp_1 \Bb (\Ff_{-})$ has the form given in (\ref{eq-invariant_states}).

{ The result follows by using the same arguments as in the proof of Theorem~\ref{theo-conv_equilibrium}.
}
\end{proof}

\begin{rmk}
	If all the operators in {\rm ({\bf Cyc})} have cyclic vectors,
	the above commutation relations imply { for instance that} $B_n^+ X_n B_n^{+}  = p (B_n^{+} \Gamma_n (V) B_n^{+}) = q ( B_n^{+} \Gamma_n (V) B_n^{-} \Gamma_n (V) B_n^{+}) =
	r ( B_n^{+} \Gamma_n (V) B_n^0 \Gamma_n (V) B_n^{+})$, where $p$, $q$ and $r$  
	In such a case, one can conlude under the additional assumption that the aforementioned equality
	implies
	$p = q =r= 1$, instead of relying on assumption {\rm ({\bf Cyc})}.
\end{rmk}


\section{Genericity of assumption {({\bf Diag})}}\label{legec}

Assumption { ({\bf Diag})} stating that { some} coefficients of the eigenvectors of $\Gamma(V)$ in the { $\{ f_j\}$-basis are non zero (see Sec.~\ref{sec-results_on_specific_model})} reflects the fact that these eigenvectors are in arbitrary positions with respect to that basis, excluding accidental symmetries. We show here that Haar distributed random unitary matrices satisfy { ({\bf Diag})} almost surely, which ensures this assumption is generically satisfied. 

\medskip

Let $C=(c_{jk})\in U(d)$ be a Haar distributed random unitary matrix on $\mathbb C^d$. For $\jv, \kv \in I_n$, $1\leq n\leq d$, consider the truncation $C_\jv (\kv )\in M_n(\mathbb C)$ of $C$

\be
C_\jv (\kv ) =
\begin{pmatrix} c_{j_1}(k_1) &c_{j_1}(k_2)& \cdots &c_{j_1}(k_n) \\
	c_{j_2}(k_1) &c_{j_2}(k_2)& \cdots &c_{j_2}(k_n)\\
	\vdots & \vdots &\ddots& \vdots\\
	c_{j_n}(k_1) &c_{j_n}(k_2)& \cdots &c_{j_n}(k_n)\end{pmatrix}\; . 
\ee
The matrix $C_\jv (\kv )$  is neither unitary, nor diagonalizable in general. However, the corresponding minor  $c_\jv (\kv )=\det(C_\jv (\kv ))$ is zero iff at least one of its eigenvalues vanishes, 
and the distribution of the eigenvalues of the bloc $C_\jv (\kv )$ is known explicitly in case $C$ is Haar distributed, see \cite{ZS, PR}.

 This allows us to get
\begin{lem}
Let $C\in U(d)$ be a random Haar distributed unitary matrix. Then, with probability one, all its minors are non-zero.
\end{lem}
\begin{proof}
For $1\leq n\leq d$,  all $n\times n$ truncations $C_\jv (\kv )$ for different $\jv, \kv \in I_n$ can be obtained as the upper left corner truncation, characterized by $\jv=\kv=(1, 2, \cdots , n)$, of the matrix  { $U_1 C U_2$, for well chosen permutation matrices $U_1, U_2\in U(d)$}. Since $C$ and $U_1 C U_2$ have the same distribution under the Haar measure, we can focus on the top left corner truncation. It is proven in  \cite{ZS, PR} that the distribution of the eigenvalues $\{z_1, z_2, \dots, z_n\}$ of that truncation has the explicit continuous joint probability density with respect to the Lebesgue measure on $\mathbb D^n\subset \mathbb C^n$
 \begin{equation}
 \frac{1}{\Nn_{n,d}}\prod_{1\leq i<j\leq n}|z_i-z_j|^2\prod_{1\leq i\leq n}(1-|z_i|^2)^{d-n-1},
 \end{equation}
 with normalization constant ${ \Nn_{n,d}}=\pi^n n! \prod_{j=1}^{n-1} \begin{pmatrix}d-n+j-1 \\ j\end{pmatrix}^{-1}\frac{1}{d-n+j}$.
 Hence the probability for at least one eigenvalue of the top left truncation to vanish is zero, since this is the probability of a zero measure set. Hence, the probability that the top left corner minor vanishes is equal to zero. Since there are finitely many truncations, this yields the result. 
\end{proof}



\begin{thebibliography}{99}

\bibitem[AAKV]{AAKV} A. Ambainis, D. Aharonov, J. Kempe, U. Vazirani, Quantum Walks on Graphs, {\it Proc. 33rd ACM STOC}, 50-59 (2001)

\bibitem[ABJ1]{ABJ2} J. Asch, O. Bourget, A. Joye, Dynamical Localization of the Chalker-Coddington Model far from Transition, {\it J. Stat. Phys.}, {\bf 147}, 194-205 (2012).


\bibitem[ABJ2]{ABJ3}  J. Asch, O. Bourget, A. Joye, Spectral Stability of Unitary Network Models, {\it Rev. Math. Phys.}, {\bf 27}, 1530004, (2015).


\bibitem[APSS]{APSS1}
S Attal, F Petruccione, C. Sabot, I. Sinayskiy, Open quantum random walks,
{\it J. Stat. Phys.} {\bf 147} (4), 832-852, (2012).


 \bibitem[AJR]{Anjora21} S. Andr\'eys, A. Joye, R. Raqu\'epas, {\it Fermionic walker driven out of equilibrium},
   J. Stat. Phys. {\bf 184}, 14 (2021)
   
 \bibitem[ASW]{ASW} A. Ahlbrecht, V. B. Scholz, A. H. Werner, Disordered quantum walks in one lattice dimension, {\it J. Math. Phys.} {\bf 52} (2011).


\bibitem[BR]{Bratteli} O. Bratteli and D.W. Robinson,
{\it Operator Algebras and Quantum Statistical Mechanics} 
(Springer, Berlin, 1997), Vol. 2

{
\bibitem[BP]{Breuer-Petruccione} H-P Breuer and F. Petruccione, The Theory of Open Quantum Systems (Oxford University Press, 2002)
}

\bibitem[BJM]{BJM} L. Bruneau, A. Joye, M. Merkli,  Repeated Interactions in Open Quantum Systems,
{\it J. Math. Phys.}, {\bf 55}, 075204, (2014). Special Issue: Non-Equilibrium Statistical Mechanics.

\bibitem[CC]{CC} J. T. Chalker, P. D.  Coddington, Percolation, quantum tunnelling and the integer Hall effect. {\it J. Phys. C: Solid State Physics}, {\bf 21}, 2665, (1988).

\bibitem[C]{C} A. M. Childs, Universal Computation by Quantum Walk, {\it Phys. Rev. Lett.} {\bf 102}, 180501 (2009).

\bibitem[CFGW]{CFGW} C. Cedzich, J. Fillman, T. Geib, A.H. Werner, Singular continuous Cantor spectrum for magnetic quantum walks
{\it  Lett. Math. Phys.}, {\bf 110} 1141-1158, (2020) .

\bibitem[CFO]{CFO} C. Cedzich, J. Fillman, D. C. Ong, Almost Everything About the Unitary Almost Mathieu Operator
{\it  Commun.Math.Phys.} {\bf 403} , 745-794, (2023)

\bibitem[CJWW]{CJWW}  C. Cedzich, A. Joye, A. H. Werner, R. F. Werner, Exponential Tail Estimates for Quantum Lattice Dynamics ", 
{\it arXiv}, arXiv:2408.02108,  (2024)

{
\bibitem[CTDRG]{Cohen-Tannoudji} C. Cohen-Tannoudji, J. Dupont-Roc, and G. Grynberg, 
Photons and atoms: Introduction to Quantum Electrodynamics (Wiley-VCH, 2004)
}
 
\bibitem[GZ]{GZ} C. Godsil, H. Zhan, Discrete quantum walks on graphs and digraphs, London Math. Soc., LNS 484, 2023

\bibitem[G]{G} S. Gudder. Quantum Markov chains. {\it J. Math. Phys.}, {\bf 49} 072105, (2008).

\bibitem[HJ1]{HJ2} E. Hamza and  A. Joye,  Spectral Properties of Non-Unitary Band Matrices.
 {\it Ann. H. Poincar\'e}  {\bf 16}, 2499-2534, (2015).

\bibitem[HJ2]{HJ} E. Hamza, A. Joye, Spectral Transition for Random Quantum Walks on Trees,
{\it Commun. Math. Phys.},  {\bf 326}, 415-439, (2014).

\bibitem[HJ3]{HamzaJoye17} E. Hamza, A. Joye, {Thermalization of Fermionic Quantum Walkers},
  {\it J. Stat. Phys.} {\bf 166}, 1365-1392, (2017)
  
\bibitem[HJS]{HJS2} E. Hamza, A. Joye, G. Stolz,  Dynamical Localization for Unitary Anderson Models,
{\it Mathematical Physics, Analysis and Geometry}, {\bf 12}, (2009), 381-444.

\bibitem[J1]{J3} A. Joye: Dynamical Localization for $d$-Dimensional Random Quantum Walks,
 {\it Quantum Inf. Proc.},  Special Issue: Quantum Walks, {\bf 11}, 1251-1269,  (2012).

\bibitem[J2]{J4} A. Joye, Dynamical Localization of Random Quantum Walks on the Lattice. In
{\it XVIIth International Congress on Mathematical Physics, Aalborg, Denmark, 6-11
August 2012}, A. Jensen, Edt., World Scientific (2013) 486-494.

\bibitem[J3]{J5}  A. Joye: Unitary and Open Scattering Quantum Walks on Graphs, {\it arXiv}, 	arXiv:2409.08428,  (2024)

\bibitem[JM]{JM} A. Joye and M. Merkli: Dynamical Localization of Quantum Walks in Random Environments,
{\it J. Stat. Phys.}, {\bf 140}, 1025-1053, (2010).

\bibitem[Ka]{Ka} T. Kato, {\it Perturbation theory for linear operators}. Springer-Verlag, New York, 1982

\bibitem[Ke]{Ke} J. Kempe, Quantum random walks: an introductory overview, {\it Contemp. Phys.}, {\bf 44}, 307–327, (2003).
 
\bibitem[Ko]{Ko} I. A. Koshovets, Unitary Analog of the Anderson
Model. Purely Point Spectrum, {\it Theoret. and Math. Phys.} {\bf
89}, 1249--1270 (1992).

{
\bibitem[Kr]{Kraus70} K. Kraus, {General State Changes in Quantum Theory}, {\it Ann. Phys.} {\bf 64}, 311-335 (1971) 
}

\bibitem[M]{Merkli06} M. Merkli, { The Ideal Quantum Gas}, in {\it ``Open Quantum Systems", Volume I: The Hamiltonian Approach},
  Springer Lecture Notes in Mathematics {\bf 1880} (2006).
  
\bibitem[Pe]{Pe} D. Petz, { \it An invitation to the algebra of canonical commutation relations}, Leuven Lecture Notes in Mathematical and Theoretical Physics, Vol. 2, Leuven University Press.

\bibitem[PR]{PR} D. Petz, J. R\'effy, { Large deviation for the empirical eigenvalue density of truncated Haar unitary matrices.} { \it Probab. Theory Relat. Fields}, {\bf 133}, 175–189 (2005)

\bibitem[P]{P} R. Portugal, {\it Quantum Walks and Search Algorithms}, Springer 2013. 
 
\bibitem[QMS]{QMS} X. Qiang, S. Ma, H. Song, Review on Quantum Walk Computing: Theory, Implementation, and Application, {\it ArXiv} arXiv:2404.04178, (2024)

\bibitem[R]{Raquepa20} R. Raqu\'epas, {On Fermionic walkers interacting with a correlated structured environment},
 { \it Lett. Math. Phys.},  {\bf 110}, 121-145 (2020)
 
\bibitem[RT]{RT} S. Richard, R. Tiedra de Aldecoa, Decay Estimates for Unitary Representations with Applications to Continuous and Discrete Time Models, {\it Ann. H. Poincar\'e}, {\bf 24}, 1-36, (2022) 

\bibitem[Sa]{Sa} M. Santha, Quantum walk based search algorithms. In {\it International Conference on Theory and Applications of Models of Computation, TAMC 2008}, 31–46. Springer, (2008).

\bibitem[Si]{Si} B. Simon, {\it Operator Theory, A Comprehensive Course in Analysis, Part 4}, AMS, 2015.

\bibitem[SFBK]{SNF} B. Sz.-Nagy, C. Foias, H. Berkovici, L. K\'erchy: {\it Harmonic Analysis of Operators in Hilbert Spaces}. Springer 2010.
  
\bibitem[T]{T} R. Tiedra de Aldecoa, Spectral and scattering properties of quantum walks on homogenous trees of odd degree,
{\it Ann. Henri Poincar\'e},  {\bf 22} , 2563-2593, (2021)
  
\bibitem[V-A]{V-A} S. E. Venegas-Andraca, Quantum walks: A comprehensive review, {\it Quantum Information Processing}, {\bf 11}, 1015–1106 (2012).
  
\bibitem[ZS]{ZS} K. Zyczkowski and H-J. Sommers,  {Truncation of random unitary matrices}, {\it J.Phys. A: Math. Gen.}, {\bf 33}, 2045–2057 (2000).
   
\end{thebibliography}
\end{document}